\documentclass[lettersize,journal]{IEEEtran}
\usepackage{amsmath,amsfonts}
\usepackage{algorithm}
\usepackage{algpseudocode}
\usepackage{array}
\usepackage[caption=false,font=normalsize,labelfont=sf,textfont=sf]{subfig}
\usepackage{textcomp}
\usepackage{stfloats}
\usepackage{url}
\usepackage{verbatim}
\usepackage{graphicx}
\def\BibTeX{{\rm B\kern-.05em{\sc i\kern-.025em b}\kern-.08em
    T\kern-.1667em\lower.7ex\hbox{E}\kern-.125emX}}

\usepackage{balance}

\usepackage[utf8]{inputenc}
\usepackage[english]{babel}

\IEEEoverridecommandlockouts 
\usepackage{array}
\newcolumntype{P}[1]{>{\centering\arraybackslash}p{#1}}

\usepackage{float}

\usepackage{graphicx}
\usepackage{blindtext}
\usepackage{epsfig}

\usepackage{mathptmx}
\usepackage{times}
\usepackage{amssymb}
\usepackage{multirow}

\usepackage{csquotes}
\usepackage{amsmath,amsthm,amssymb}

\newtheorem{theorem}{Theorem}

\usepackage{tikz}
\usetikzlibrary{shapes.geometric, arrows}
\usepackage{longtable}

\usepackage{amsbsy}

\usepackage{authblk}
\usepackage{leftidx}
\usepackage{ragged2e}
\usepackage[font={scriptsize}]{caption}
\usepackage{color}

\makeatletter
\newcommand*\titleheader[1]{\gdef\@titleheader{#1}}
\AtBeginDocument{%
	\let\st@red@title\@title
	\def\@title{%
		\bgroup\normalfont\large\centering\@titleheader\par\egroup
		\vskip1.5em\st@red@title}
}
\makeatother

\titleheader{This work has been submitted to the IEEE for possible publication. Copyright may be transferred without notice, after which this version may no longer be accessible.}

\title{Secure Event-Triggered Distributed Kalman Filters for State Estimation over Wireless Sensor Networks}

\begin{document}

\author{Aquib Mustafa, Majid Mazouchi, Hamidreza Modares, \textit{Senior Member, IEEE} 
	\thanks{
		Aquib Mustafa, Majid Mazouchi and Hamidreza Modares are with  the Department
		of Mechanical Engineering, Michigan State University, East Lansing, MI, 48863, USA (e-mails: mustaf15@msu.edu; 
		mazouchi@msu.edu; modaresh@msu.edu).}
}

\maketitle

\begin{abstract}
	In this paper, we analyze the adverse effects of cyber-physical attacks as well as mitigate their impacts on the event-triggered distributed Kalman filter (DKF). We first show that although event-triggered mechanisms are highly desirable, the attacker can leverage the event-triggered mechanism to cause non-triggering misbehavior which significantly harms the network connectivity and its collective observability. We also show that an attacker can mislead the event-triggered mechanism to achieve continuous-triggering misbehavior which not only drains the communication resources but also harms the network's performance. An information-theoretic approach is presented next to detect attacks on both sensors and communication channels. {In contrast to the existing results, the restrictive Gaussian assumption on the attack signal's probability distribution is not required.} To mitigate attacks, a meta-Bayesian approach is presented that incorporates the outcome of the attack detection mechanism to perform second-order inference. The proposed second-order inference forms confidence and trust values about the truthfulness or legitimacy of sensors' own estimates and those of their neighbors, respectively. Each sensor communicates its confidence to its neighbors. Sensors then incorporate the confidence they receive from their neighbors and the trust they formed about their neighbors into their posterior update laws to successfully discard corrupted information. Finally, the simulation result validates the effectiveness of the presented resilient event-triggered DKF.
\end{abstract}


\section{Introduction}
{Cyber-physical systems (CPSs) refer to a class of engineering systems that integrate the cyber aspects of computation and communication with physical entities \cite{c1}.} Integrating communication and computation with sensing and control elements has made CPSs a key enabler in designing emerging autonomous and smart systems with the promise of bringing unprecedented benefits to humanity. CPSs have already had a profound impact on variety of engineering sectors, including, process industries \cite{c2}, robotics \cite{c3}, {smart grids \cite{c4}, and intelligent transportation \cite{c5}, health care system \cite{ss1}, to name a few.} Despite their advantages with vast growth and success, these systems are vulnerable to cyber-physical threats and can face fatal consequences if not empowered with resiliency. {The importance of designing resilient and secure CPSs can be witnessed from severe damages made by recently reported cyber-physical attacks \cite{c6_7}}.\par 

\subsection{Related Work}
Wireless sensor networks (WSNs) are a class of CPSs for which a set of sensors are spatially distributed to monitor and estimate a variable of interest (e.g., location of a moving target, state of a large-scale system, etc.), {and have various applications such as surveillance and monitoring, target tracking, and active health monitoring \cite{c8}}. In centralized WSNs, all sensors broadcast their measurements to a center at which the information is fused to estimate the state \cite{c9,ss3}. These approaches, however, are communication demanding and prone to single-point-of-failure. {To estimate the state with reduced communication burden, a distributed Kalman filter (DKF) is presented in \cite{c10}-\cite{d3}, in which sensors exchange their information only with their neighbors, not with all agents in the network or a central agent.} 
{Cost constraints on sensor nodes in a WSN result in corresponding constraints on resources such as energy and communications bandwidth. Sensor nodes in a WSN usually carry limited, irreplaceable energy resources and lifetime adequacy is a significant restriction of almost all WSNs.  Therefore, it is of vital importance to design event-triggered DKF to reduce the communication burden which consequently improves energy efficiency. To this end, several energy-efficient event-triggered distributed state estimation approaches are presented for which sensor nodes intermittently exchange information \cite{c13}-\cite{c16}. Moreover, the importance of event-triggered state estimation problem is also reported for several practical applications such as smart grids and robotics \cite{r02}-\cite{r04}.}
Although event-triggered distributed state estimation is resource-efficient, it provides an opportunity for an attacker to harm the network performance and its connectivity by corrupting the information that is exchanged among sensors, as well as to mislead the event-triggered mechanism. Thus, it is of vital importance to design a resilient event-triggered distributed state estimation approach that can perform accurate state estimation despite attacks. \par 

In recent years, secure estimation and secure control of CPSs have received significant attention and remarkable results have been reported for mitigation of cyber-physical attacks, {including denial of service (DoS) attacks \cite{c17}-\cite{c18}, false data injection attacks \cite{c19}-\cite{c23}, and bias injection attacks \cite{c24}. }For the time-triggered distributed scenario, several secure state estimation approaches are presented in \cite{c26}-\cite{c312}. Specifically, in \cite{c26}-\cite{c30} authors presented a distributed estimator that allows agents to perform parameter estimation in the presence of attack by discarding information from the adversarial agents. Byzantine-resilient distributed estimator with deterministic process dynamics is discussed in \cite{c27}. Then, the same authors solved the resilient distributed estimation problem with communication losses and intermittent measurements in \cite{c28}. Attack analysis and detection for distributed Kalman filters are discussed in \cite{c281}. Resilient state estimation subject to DoS attacks for power system and robotics applications is presented in \cite{c310}-\cite{c312}. Although meritable, these aforementioned results for the time-triggered resilient state estimation  are not applicable to event-triggered distributed state estimation problems. {Recently, authors in \cite{c17r} addressed the event-triggered distributed state estimation under DoS attacks by employing the covariance intersection fusion approach. Although elegant, the presented approach is not applicable to mitigating the effect of deception attacks. To our knowledge, resilient state estimation for event-triggered DKF under deception attacks is not considered in the literature. For the first time, this work not only detects and mitigate the effect of attacks on sensor and communication channel but also presents a mathematical analysis for different triggering misbehaviors.}

\vspace{-0.35cm}

\subsection{Contributions and outline}
\vspace{-0.1cm}

{This paper contributes to analysis, detection, and mitigation of attacks on event-triggered DKF. {To our knowledge, it is the first paper to rigorously analyze how an attacker can leverage the event- triggering mechanism to harm the state estimation process over WSNs. It also proposes a novel detection mechanism for detecting attacks on event-triggered DKF that does not require the restrictive Gaussian assumption on the probability density function of the attack signal. To provide mitigation scheme and discard corrupted information, finally, a novel meta-Bayesian mechanism is developed that performs second-order inference to form confidence and trust about the truthfulness or legitimacy of the outcome of its own first-order inference and those of its neighbors, respectively.} The details of these contributions are presented as follows:}
{\begin{itemize}
		\item {Attack analysis: It is shown that the attacker can cause emerging non-triggering misbehavior so that the compromised sensors do not broadcast any information to their neighbors. This can significantly harm the network connectivity and its collective observability, which is the necessary condition for solving the distributed state estimation problem. It is also shown that an attacker can achieve continuous-triggering misbehavior which drains the communication resources.}
		\item {Attack detections: To detect adversarial intrusions a Kullback-Leibler (KL) divergence based detector is presented and estimated via k-nearest neighbors approach to obviate the restrictive Gaussian assumption on the probability density function of the attack signal.}
		\item {Attack mitigation: To mitigate attacks on event-triggered DKF and discard corrupted information, a meta-Bayesian approach is employed that performs second-order inference to form confidence and trust about the truthfulness or legitimacy of the outcome of its own first-order inference (i.e., the posterior belief about the state estimate) and those of its neighbors, respectively. Each sensor communicates its confidence to its neighbors and also incorporates the trust about its neighbors into its posterior update law to put less weight on untrusted data and thus successfully discard corrupted information.}
\end{itemize}}

\textit{Outline:} The paper is organized as follows. Section II outlines the preliminary background for the event-triggered DKF. Section III formulates the effect of attacks on the event-triggered DKF and analyzes triggering misbehaviors for it. Attack detection mechanism and confidence-trust based secure event-triggered DKF are presented in Section IV and V, respectively. The simulation verifications are provided in Section VI. Finally, concluding remarks are presented in Section VII.

\vspace{-0.1cm}
\section{Notations and Preliminaries}

\subsection{Notations}
The data communication among sensors in a WSN is captured by an undirected graph ${\rm {\mathcal G}}$, consists of a pair $({\rm {\mathcal V}},{\rm {\mathcal E}})$, where ${\rm {\mathcal V}}=\{ 1,2,\ldots ,N\}$ is the set of nodes or sensors and ${\rm {\mathcal E}}\subset {\rm {\mathcal V}}\times {\rm {\mathcal V}}$ is the set of edges. An edge from node $j$ to node $i,$ represented by $(j,i)$, implies that node $j$ can broadcast information to node $i$. Moreover, $N_{i} =\{ j:(j,i)\in {\rm {\mathcal E}}\}$ is the set of neighbors of node $i$ on the graph ${\rm {\mathcal G}}.$ An induced subgraph ${\rm {\mathcal G}}^{w}$ is obtained by removing a set of nodes ${\rm {\mathcal W}}\subset {\rm {\mathcal V}}$ from the original graph ${\rm {\mathcal G}}$, which is represented by nodes set ${\rm {\mathcal V}\backslash {\mathcal W}}$ and contains the edges of  ${\rm {\mathcal E}}$ with both endpoints in ${\rm {\mathcal V}\backslash {\mathcal W}}$. 

Throughout this paper, ${\bf {\mathbb{R}}}$ and ${\bf {\mathbb{N}}}$ represent the sets of real numbers and natural numbers, respectively. $A^{T}$ denotes transpose of a matrix $A$. $tr(A)$ and $\max (a_{i} )$ represent trace of a matrix $A$ and maximum value in the set, respectively. ${\rm {\mathcal C}}(S)$ represents the cardinality of a set S. $\sigma _{\max } (A),$ $\lambda _{\max } (A),$ and $I_{n}$ represent maximum singular value, maximum eigenvalue of matrix A, and an identity matrix of dimension $n$, respectively. ${\rm {\mathcal U}}(a,b)$ with $a<b$ denotes an uniform distribution between the interval $a$ and $b$. Consider $p_{X} (x)$ as the probability density of the random variable or vector $x$ with $X$ taking values in the finite set $\{ 0,...,p\}.$ When a random variable $X$ is distributed normally with mean $\nu$ and variance $\sigma ^{2},$ we use the notation $X\sim {\rm {\mathcal N}}(\upsilon ,\sigma ^{2} )$. ${\bf {\rm E}}[X]$ and $\Sigma _{X} ={\bf {\rm E}}[(X-{\bf {\rm E}}[X])(X-{\bf {\rm E}}[X])^{T} ]$ denotes, respectively, the expectation and the covariance of $X.$ Finally, ${\bf {\rm E}}[.|.]$ represents the conditional expectation. 

\vspace{-0.3cm}
\subsection{Process Dynamics and Sensor Models}
Consider a process that evolves according to the following dynamics
\vspace{-0.15cm}
\begin{equation} \label{ZEqnNum820040} 
	x(k+1)=Ax(k)\, +\, w(k), 
\end{equation} 
where $A$ denotes the process dynamic matrix, and $x(k)\in {\bf {\mathbb R}}^{n}$ and $w(k)$ are, respectively, the process state and process noise at the time $k$. The process noise $w(k)$ is assumed to be independent and identically distributed (i.i.d.) with Gaussian distribution, and $x_{0} \in {\rm {\mathcal N}}(\hat{x}_{0} ,P_{0} )\,$ represents the initial process state with $\hat{x}_{0}$ as mean and $P_{0}$ as covariance, respectively.

The goal is to estimate the state $x(k)$ for the process \eqref{ZEqnNum820040} in a distributed fashion using $N$ sensor nodes that communicate through the graph ${\rm {\mathcal G}}$, and their sensing models are given by
\begin{equation} \label{ZEqnNum687942} 
	y_{i} (k)=C_{i} x_{i} (k)\, +\, v_{i} (k);\, \, \, \, \, \, \, \, \, \, \, \forall i=1,\cdots ,N, 
\end{equation} 
where $y_{i} (k)\in {\bf {\mathbb R}}^{p}$ represents the measurement data with $v_{i} (k)$ as the i.i.d. Gaussian measurement noise and $C_{i}$ as the observation matrix of the sensor $i$, respectively.

\smallskip

\noindent
\textbf{Assumption 1}. The process noise $w(k),$ the measurement noise $v_{i} (k),$ and the initial state $x_{0}$ are uncorrelated random vector sequences. 

\smallskip

\noindent 
\textbf{Assumption 2}. The sequences $w(k)$ and $v_{i}(k)$ are zero-mean Gaussian noise with
\vspace{-0.15cm}
\[{\bf {\rm E}}[w(k)(w(h))^{T} ]=\mu _{kh} Q\, \, \, \, \] 
and 
\vspace{-0.15cm}
\[{\bf {\rm E}}[v_{i} (k)(v_{i} (h))^{T} ]=\mu _{kh} R_{i} ,\] 
with $\mu _{kh} =0$ if $k\ne h$,  and $\mu _{kh} =1$ otherwise. Moreover, $Q\ge0$ and $R_{i}>0$ denote the noise covariance matrices for process and measurement noise, respectively and both are finite. 

\smallskip

\noindent 
\textbf{Definition 1. (Collectively observable) \cite{c11}.} We call the plant dynamics \eqref{ZEqnNum820040} and the measurement equation \eqref{ZEqnNum687942} collectively observable, if the pair $(A,C_{S} )$ is observable where $C_{s}$ is the stack column vectors of $C_{j}, \,\,\forall j \in S$ with $S\subseteq {\rm {\mathcal V}}$ and ${\rm {\mathcal C}}(S)>N/2$. 

\smallskip

\noindent 
\textbf{Assumption 3.} The plant dynamics \eqref{ZEqnNum820040} and the measurement equation \eqref{ZEqnNum687942} are collectively observable, but not necessarily locally observable, i.e., $(A,C_{i} )$ $\, \forall i\in {\rm {\mathcal V}}$ is not necessarily observable. 

Assumptions $1$ and $2$ are standard assumptions in Kalman filters. {Assumption 3  states that the state of the target in \eqref{ZEqnNum820040} cannot be observed by measurements of any single sensor, i.e., the pairs $(A,C_{i} )$ cannot be observable (see for instances \cite{c11} and \cite{c30}). It also provides the necessary assumption of collectively observable for the estimation problem to be solvable. Also note that under Assumption 2, i.e., the process and measurement covariance are finite, the stochastic observability rank condition coincides with the deterministic observability [Theorem 1, 43]. Therefore, deterministic observability rank condition holds true irrespective of the process and measurement noise.} 

\vspace{-0.3cm}
\subsection{Overview of Event-triggered Distributed Kalman Filter}
This subsection presents the overview of the event-triggered DKF for estimating the process state $x(k)$ in \eqref{ZEqnNum820040} from a collection of noisy measurements $y_{i} (k)$ in \eqref{ZEqnNum687942}.

Let the prior and posterior estimates of the target state $x(k)$ for sensor node $i$ at time $k$ be denoted by $x_{i}(k|k-1)$ and $x_{i}(k|k)$, respectively. In the centralized Kalman filter, a recursive rule based on Bayesian inference is employed to compute the posterior estimate $x_{i}(k|k)$ based on its prior estimate $x_{i}(k|k-1)$ and the new measurement $y_{i}(k)$. When the next measurement comes, the previous posterior estimate is used as a new prior and it proceeds with the same recursive estimation rule. In the event-triggered DKF, the recursion rule for computing the posterior incorporates not only its own prior and observations, but also its neighbors' predictive state estimate. Sensor $i$ communicates its prior state estimate to its neighbors and if the norm of the error between the actual output and the predictive output becomes greater than a threshold after a new observation arrives. That is, it employs the following event-triggered mechanism for exchange of data with its neighbors
\vspace{-0.15cm}
\begin{equation} \label{eq3x}
	\left\| y_{i} (k)-C_{i} \tilde{x}_{i} (k-1)\right\| <\alpha, 
\end{equation}
where $\alpha$ denotes a predefined threshold for event-triggering. Moreover, $\tilde{x}_{i} (k)$ denotes the predictive state estimate for sensor $i$ and follows the update law
\begin{equation} \label{ZEqnNum926700} 
	\tilde{x}_{i} (k)=\zeta _{i} (k)x_{i} (k|k-1)+(1-\zeta _{i} (k))A\tilde{x}_{i} (k-1),\, \, \forall i\in {\rm {\mathcal V}}, 
\end{equation} 
with $\zeta _{i} (k)\in \left\{0,1\right\}$ as the transmit function. {Note that the predictive state estimate update equation in (4) depends on the value of the transmit function ${{\zeta }_{i}}(k)$ which is either zero or one depending on the triggering condition in (3). When ${{\zeta }_{i}}(k)=1$, then the prior and predictive state estimates are the same, i.e., ${{\tilde{x}}_{i}}(k)={{x}_{i}}(k|k-1)$.  When ${{\zeta }_{i}}(k)=0,$ however, the predictive state estimate depends on its own previous state estimate, i.e., ${{\tilde{x}}_{i}}(k)=A{{\tilde{x}}_{i}}(k-1).$ }

Incorporating \eqref{ZEqnNum926700}, the following recursion rule is used to update the posterior state estimate in the event-triggered DKF \cite{c13}, \cite{c15} for sensor $i$ as
\begin{equation} \label{ZEqnNum257073} 
	\begin{array}{l} {x_{i} (k|k)=x_{i} (k|k-1)+K_{i} (k)(y_{i} (k)-C_{i} x_{i} (k|k-1))} \\ {\, \, \, \, \, \, \,\, \, \, \, \, \,\, \, \, \, \, \, \,  \,\, \, \, \, \, \, \, \, \, \, \, \, +\gamma _{i} \sum _{j\in N_{i} }(\tilde{x}_{j} (k)-\tilde{x}_{i} (k) ),} \end{array} 
\end{equation} 
where 
\vspace{-0.25cm}
\begin{equation} \label{ZEqnNum569383} 
	x_{i} (k|k-1)=Ax_{i} (k-1|k-1), 
\end{equation} 
is the prior update. Moreover, the second and the third terms in \eqref{ZEqnNum257073} denote, respectively, the innovation part (i.e., the estimation error based on the sensor $i^{th}$ new observation and its prior prediction) and the consensus part (i.e., deviation of the sensor state estimates from its neighbor's state estimates). We call this recursion rule as the \textit{Bayesian first-order inference} on the posterior, which provides the belief over the value of the state. 

Moreover, $K_{i} (k)$ and $\gamma _{i}$ in \eqref{ZEqnNum257073}, respectively, denote the Kalman gain and the coupling coefficient. The Kalman gain $K_{i} (k)$ in \eqref{ZEqnNum257073} depends on the estimation error covariance matrices associated with the prior $x_{i} (k|k-1)$ and the posterior $x_{i} (k|k)$ for the sensor $i$. Let define the prior and posterior estimated error covariances as
\begin{equation} \label{ZEqnNum606287} 
	\begin{array}{l} {P_{i} (k|k-1)={\bf {\rm E}}[(x(k)-x_{i} (k|k-1))(x(k)-x_{i} (k|k-1))^{T} ],} \\ {P_{i} (k|k)={\bf {\rm E}}[(x(k)-x_{i} (k|k))(x(k)-x_{i} (k|k))^{T} ].} \end{array} 
\end{equation} 
which are simplified as \cite{c13}, \cite{c15}
\begin{equation} \label{ZEqnNum987927} 
	P_{i} (k|k)=M_{i} (k)P_{i} (k|k-1)(M_{i} (k))^{T} +K_{i} (k)R_{i} (K_{i} (k))^{T} , 
\end{equation} 
and
\vspace{-0.25cm}
\begin{equation} \label{9)} 
	P_{i} (k|k-1)=AP_{i} (k-1|k-1)A^{T} +Q. 
\end{equation} 
with  $M_{i} (k)=I_{n} -K_{i} (k)C_{i}.$ Then, the Kalman gain $K_{i} (k)$ is designed to minimize the estimation covariance and is given by \cite{c13}, \cite{c15}
\begin{equation} \label{ZEqnNum999982} 
	K_{i} (k)=P_{i} (k|k-1)(C_{i} )^{T} (R_{i} (k)+C_{i} P_{i} (k|k-1)(C_{i} )^{T} )^{-1} . 
\end{equation} \par
Let the innovation sequence $r_{i} (k)$ for the node $i$ be defined as
\begin{equation} \label{ZEqnNum276515} 
	r_{i} (k)=y_{i} (k)-C_{i} x_{i} (k|k-1), 
\end{equation} 
\vspace{-0.15cm}
where $r_{i}(k)\sim {\rm {\mathcal N}}(0,\Omega _{i} (k))$ with
\begin{equation} \label{ZEqnNum368934} \nonumber
	\Omega _{i} (k)={\bf {\rm E}}[r_{i} (k)(r_{i} (k))^{T} ]=C_{i} P_{i} (k|k-1)C_{i} {}^{T} +R_{i} (k). 
\end{equation}\par
Note that for the notional simplicity, henceforth we denote the prior and posterior state estimations as $x_{i} (k|k-1)\buildrel\Delta\over= \bar{x}_{i} (k)$ and $x_{i} (k|k)\buildrel\Delta\over= \hat{x}_{i} (k),$ respectively. Also, the prior covariance and the posterior covariance are, respectively, denoted by $P_{i} (k|k-1)\buildrel\Delta\over= \bar{P}_{i} (k)$ and $P_{i} (k|k)\buildrel\Delta\over= \hat{P}_{i} (k)$. \par

\smallskip

{Based on equations (6)-(10)}, the event-triggered DKF algorithm becomes 

$\textit{Time\,\,updates:}$ \par \hfill
\vspace{-0.4cm}
\begin{equation}\label{ZEqnNum838493}
	\left\{ {\begin{array}{*{20}{c}}
			\bar{x}_{i}(k+1)=A{{{\hat{x}}}_{i}(k)}\,\,\,\,\,\,\,\,\,\,\,\,\,\,\,\,\,\,\,\,\,\,\,\,\,\,\,\,\,\,\,\,\,\,\,\,\,\,\,\,\,\,\,\,\,\,\,\,\,\,\,\,\,\,\,\,\,\,\,\,\,\,\,\,\,\,\,\,\,\,\,\,\,\,\,\,\,\,\,\,\,\,\,\,\,\,\,\,\,\,\,(a)   \\ 
			\bar{P}_{i}(k+1)=A{{{\hat{P}}}_{i}}(k){{A}^{T}}+Q(k)\,\,\,\,\,\,\,\,\,\,\,\,\,\,\,\,\,\,\,\,\,\,\,\,\,\,\,\,\,\,\,\,\,\,\,\,\,\,\,\,\,\,\,\,\,\,\,\,\,\,\,\,\,\,\,(b) 
	\end{array}} \right.
\end{equation}\par

$\textit{Measurment\,\,updates:}$\par 
\vspace{-0.4cm}
\begin{equation}\label{ZEqnNum727229}
	{\left\{\begin{array}{l} {\hat{x}_{i} (k)=\bar{x}_{i} (k)+K_{i} (k)(y_{i} (k)-C_{i} \bar{x}_{i} (k))} \\ {\, \, \, \, \, \, \, \, \, \, \, \, \, \,  \,\, \, \, \, \, \, \, \, \, \, \, \, +\gamma _{i} \sum _{j\in N_{i} }(\tilde{x}_{j} (k)-\tilde{x}_{i} (k) ),\, \, \, \, \, \, \, \, \, \, \, \, \, \, \, \, \, \, \,\, \, \, \, \, \, \, \, \, \, \, \, \, \, \, \, \, \, \, \, \, \, \, \, \, \, \, \, \,  (a)} \\ {\tilde{x}_{i} (k)=\zeta _{i} (k)\bar{x}_{i} (k)+(1-\zeta _{i} (k))A\tilde{x}_{i} (k-1),\, \, \, \, \, \, \, \, \, \, \, \, \, (b)} \\ {K_{i} (k)=\bar{P}_{i} (k)C_{i}^{T} (R_{i} (k)+C_{i} \bar{P}_{i} (k)C_{i}^{T} )^{-1} ,\, \, \, \, \, \, \, \, \, \, \,  \, \, \, \, \, \, \, \,\, \, \, \, (c)} \\ {\hat{P}_{i} (k)=M_{i} \bar{P}_{i} (k)M_{i} {}^{T} +K_{i} (k)R_{i} (k)(K_{i} (k))^{T} .\, \, \, \, \, \, \,\, \, \, \, \, \, \,  (d)} \end{array}\right. }
\end{equation}

\smallskip
\noindent
\noindent \textbf{Remark 1.} Based on the result presented in [17, Th.1], the event triggered DKF \eqref{ZEqnNum838493}-\eqref{ZEqnNum727229} ensures that the estimation error $\hat{x}_{i} (k)-x(k)$ is exponentially bounded in the mean square sense $\forall i\in {\rm {\mathcal V}}.$

\smallskip

\noindent
\noindent \textbf{Remark 2.} {The consensus gain ${{\gamma }_{i}}$ in (5) is designed such that the stability of the event-triggered DKF in (13)-(14) is guaranteed. Specifically, as shown in [Theorem 2, 19], if  
	\begin{equation}
		\nonumber
		{{\gamma }_{i}}=\frac{2(I-{{K}_{i}}{{C}_{i}}){{({{\Gamma }_{i}})}^{-1}}}{{{\lambda }_{\max }}(\mathcal{L}){{\lambda }_{\max }}({{(\Gamma )}^{-1}})}
	\end{equation}
	where $\mathcal{L}$ denotes the Laplacian matrix associated with the graph $\mathcal{G}$ and $\Gamma =diag\{{{\Gamma }_{1}},..,{{\Gamma }_{N}}\}$ with ${{\Gamma }_{i}}={{(I-{{K}_{i}}{{C}_{i}})}^{T}}{{A}^{T}}{{({{\bar{P}}_{i}})}^{+}}A(I-{{K}_{i}}{{C}_{i}}),\,\,\forall i=\{1,...,N\},$ then the stability of the event-triggered DKF in (13)-(14) is guaranteed. However, the design of event-triggered DKF itself is not the concern of this paper and this paper mainly analyzes the adverse effects of cyber-physical attacks on the event-triggered DKF and proposes an information-theoretic approach based attack detection and mitigation mechanism. Note that the presented attack analysis and mitigation can be extended to other event-triggered methods such as \cite{c14} and \cite{c16} as well.}

\vspace{-0.3cm}
\subsection{Attack Modeling}
In this subsection, we model the effects of attacks on the event-triggered DKF. An attacker can design a false data injection attack to affect the triggering mechanism presented in (\ref{eq3x}) and consequently compromise the system behavior. 

\smallskip

\noindent
\textbf{Definition 2. (Compromised and intact sensor node).} We call a sensor node that is directly under attack as a compromised sensor node. A sensor node is called intact if it is not compromised. Throughout the paper, ${\rm {\mathcal V}}^{c}$ and ${\rm {\mathcal V}}\backslash {\rm {\mathcal V}}^{c}$ denote, respectively, the set of compromised and intact sensor nodes. 

\smallskip

Consider the sensing model \eqref{ZEqnNum687942} for sensor node $i$ under the effect of the attack as
\begin{equation} \label{ZEqnNum973066} 
	y_{i}^{a} (k)=y_{i} (k)+f_{i} (k)=C_{i} x_{i} (k)\, +\, v_{i} (k)+f_{i} (k), 
\end{equation} 
where $y_{i} (k)$ and $y_{i}^{a}(k)$ are, respectively, the sensor $i$'$s$ actual and corrupted measurements and $f_{i} (k)\in {\bf {\rm R}}^{p}$ represents the adversarial input on sensor node $i.$ For a compromised sensor node $i,$ let $p'\subseteq p$ be the subset of measurements disrupted by the attacker.\par

Let the false data injection attack $\bar{f}_{j}(k)$ on the communication link be given by 
\vspace{-0.2cm}
\begin{equation} \label{ZEqnNum397788} 
	\bar{x}_{j}^{a} (k)=\bar{x}_{j} (k)+\bar{f}_{j} (k),\, \, \, \forall j\in N_{i} . 
\end{equation} 

Using \eqref{ZEqnNum973066}-\eqref{ZEqnNum397788}, in the presence of an attack on sensor node $i$ and/or its neighbors, its state estimate equations in \eqref{ZEqnNum727229}-\eqref{ZEqnNum838493} becomes
\begin{equation} \label{ZEqnNum120276} 
	\left\{\begin{array}{l} {\hat{x}_{i}^{a} (k)=\bar{x}_{i}^{a} (k)+K_{i}^{a} (k)(y_{i} (k)-C_{i} \bar{x}_{i}^{a} (k))} \\ {\, \, \, \, \, \, \, \, \, \, \, \, \, \,  \,\, \, \, \, \, \, \, \, \, \, \, \, +\gamma _{i} \sum _{j\in N_{i} }(\tilde{x}_{j} (k)-\tilde{x}_{i}^{a} (k) )+f_{i}^{a} (k),} \\ {\bar{x}_{i}^{a} (k+1)=A\hat{x}_{i}^{a} (k),} \\ {\tilde{x}_{i}^{a} (k)=\zeta _{i} (k)\bar{x}_{i}^{a} (k)+(1-\zeta _{i} (k))A\tilde{x}_{i}^{a} (k-1),} \end{array}\right.  
\end{equation} 
where\vspace{-0.15cm}
\begin{equation} \label{ZEqnNum499212} 
	f_{i}^{a} (k)=K_{i}^{a} (k)f_{i} (k)+\gamma _{i} \sum _{j\in N_{i} }\tilde{f}_{j} (k) , 
\end{equation} 
with 
\vspace{-0.15cm}
\begin{equation} \label{ZEqnNum429253} \nonumber
	\tilde{f}_{j} (k)=\zeta _{j} (k)\bar{f}_{j} (k)+(1-\zeta _{j} (k))\tilde{f}_{j} (k-1). 
\end{equation} 
The Kalman gain $K_{i}^{a} (k)$ in presence of attack is given by
\begin{equation} \label{ZEqnNum654467} 
	K_{i}^{a} (k)=\bar{P}_{i}^{a} (k)C_{i}^{T} (R_{i} (k)+C_{i} \bar{P}_{i}^{a} (k)C_{i}^{T} )^{-1} . 
\end{equation} 
The first part in \eqref{ZEqnNum499212} represents the direct attack on sensor node $i$  and the second part denotes the aggregative effect of adversarial input on neighboring sensors, i.e., $j\in N_{i}$. Moreover, $\hat{x}_{i}^{a}(k),\, \, \bar{x}_{i}^{a} (k),$ and $\tilde{x}_{i}^{a}(k)$ denote, respectively, the corrupted posterior, prior, and predictive state estimates. The Kalman gain $K_{i}^{a}(k)$ depends on the following corrupted prior state estimation error covariance 
\begin{equation} \label{ZEqnNum384197} 
	\bar{P}_{i}^{a} (k+1)=A\hat{P}_{i}^{a} (k)A^{T} +Q. 
\end{equation} 
where the corrupted posterior state estimation error covariance $\hat{P}_{i}^{a} (k)$ evolution is shown in the following theorem.

\begin{theorem}
	Consider the process dynamics \eqref{ZEqnNum820040} with compromised sensor model \eqref{ZEqnNum973066}. Let the state estimation equation be given by \eqref{ZEqnNum120276} in the presence of attacks modeled by $f_{i}^{a}(k)$ in \eqref{ZEqnNum499212}. Then, the corrupted posterior state estimation error covariance $\hat{P}_{i}^{a}(k)$ is given by 
	\begin{equation} \label{ZEqnNum998129} 
		\begin{array}{l} {\hat{P}_{i}^{a} (k)=M_{i}^{a} (k)\bar{P}_{i}^{a} (k)(M_{i}^{a} (k))^{T} +K_{i}^{a} (k)[R_{i} (k)+\Sigma _{i}^{f} (k)](K_{i}^{a} (k))^{T} } \\ \, \, \, \, \, \, \, \, \, \, \, \, \, \,  \,\, \, \, \, \, \, \, \, \, \, \, \, {+2\gamma _{i} \sum _{j\in N_{i} }(\stackrel{\frown}{P}_{i,j}^{a} (k) -\stackrel{\frown}{P}_{i}^{a} (k))(M_{i}^{a} (k))^{T} -2K_{i}^{a} (k)\Xi _{f} (k)} \\ \, \, \, \, \, \, \, \, \, \, \, \, \, \,  \,\, \, \, \, \, \, \, \, \, \, \, \,{+\gamma _{i} {}^{2} (\sum _{j\in N_{i} }(\tilde{P}_{j}^{a} (k) -2\tilde{P}_{i,j}^{a} (k)+\tilde{P}_{i}^{a} (k))}, \end{array} 
	\end{equation} 
	where $\Sigma _{i}^{f}(k)$ and $\Xi _{f} (k)$ denote the attacker's input dependent covariance matrices and $M_{i}^{a} =(I_{n} -K_{i}^{a} (k)C_{i} )$ with $K_{i}^{a} (k)$ as the Kalman gain and $\bar{P}_{i}^{a} (k)$ as the prior state estimation error covariance update in \eqref{ZEqnNum654467} and \eqref{ZEqnNum384197}, respectively. Moreover, $\tilde{P}_{i,j}^{a} (k)$ and $\stackrel{\frown}{P}_{i,j}^{a}(k)$ are cross-correlated estimation error covariances updated according to \eqref{ZEqnNum928831}-\eqref{ZEqnNum358063}.
\end{theorem}

\begin{proof}
	See Appendix A.    
\end{proof}

\vspace{-0.2cm}
Note that the corrupted state estimation error covariance recursion $\hat{P}_{i}^{a} (k)$ in \eqref{ZEqnNum998129} depends on the attacker's input distribution. Since the state estimation depends on compromised estimation error covariance $\hat{P}_{i}^{a} (k),$ therefore, the attacker can design its attack signal to blow up the estimates of the desired process state and damage the system performance.

\vspace{-0.2cm}
\section{ Effect of Attack on Triggering Mechanism}

This section presents the effects of cyber-physical attacks on the event-triggered DKF. We show that although event-triggered approaches are energy efficient, they are prone to triggering misbehaviors, which can harm the network connectivity, observability and drain its limited resources. 

\vspace{-0.35cm}

\subsection{ Non-triggering Misbehavior}
In this subsection, we show how an attacker can manipulate the sensor measurement to mislead the event-triggered mechanism and damage network connectivity and collective observability by causing \textit{non-triggering misbehavior} as defined in the following Definition 3.

\smallskip

\noindent 
\textbf{Definition 3 }(\textbf{Non-triggering Misbehavior).} The attacker designs an attack strategy such that a compromised sensor node does not transmit any information to its neighbors by misleading the triggering mechanism in (\ref{eq3x}), even if the actual performance deviates from the desired one. 

The following theorem shows how a false data injection attack, followed by an eavesdropping attack, can manipulate the sensor reading to avoid the event-triggered mechanism (\ref{eq3x}) from being violated while the actual performance could be far from the desired one. To this end, we first define the vertex cut of the graph as follows.

\smallskip

\noindent
\textbf{Definition 4 (Vertex cut).} A set of nodes ${\rm {\mathcal C}}\subset {\rm {\mathcal V}}$ is a vertex cut of a graph ${\rm {\mathcal G}}$ if removing the nodes in the set ${\rm {\mathcal C}}$ results in disconnected graph clusters.

\begin{theorem}
	Consider the process dynamics \eqref{ZEqnNum820040} with $N$ sensor nodes \eqref{ZEqnNum687942} communicating over the graph ${\rm {\mathcal G}}$. Let sensor $i$ be under a false data injection attack given by
	\begin{equation} \label{ZEqnNum705143} 
		y_{i}^{a} (k)=y_{i} (k)+\theta _{i}^{a} (k)1_{p} ,\, \, \, \, \forall k\ge L+1, 
	\end{equation} 
	where $y_{i}(k)$ is the actual sensor measurement  at time instant $k$ and $L$ denotes the last triggering time instant. Moreover, $\theta _{i}^{a}(k)\sim {\rm {\mathcal U}}(a(k),b(k))\, $ is a scalar uniformly distributed random variable in the interval $(a(k),b(k))$ with 
	\begin{equation} \label{ZEqnNum165624} 
		\left\{\begin{array}{l} {a(k)=\varphi -\left\| C_{i} \tilde{x}_{i} (k-1)\right\| +\left\| y_{i} (k)\right\|, } \\ {b(k)=\varphi +\left\| C_{i} \tilde{x}_{i} (k-1)\right\| -\left\| y_{i} (k)\right\|, } \end{array}\right.  
	\end{equation} 
	where $\tilde{x}_{i} (k)$ and $\varphi <\alpha $ denote, respectively, the predictive state estimate and an arbitrary scalar value less than the triggering threshold $\alpha .$  Then, 
	
	\begin{enumerate}
		\item  The triggering condition (\ref{eq3x})  will not be violated for the sensor node $i$ and it shows non-triggering misbehavior;
		\item  The original graph ${\rm {\mathcal G}}$ is clustered into several subgraphs, if all sensors in a vertex cut are under attack \eqref{ZEqnNum705143}.
	\end{enumerate}
\end{theorem}

\begin{proof}
	Taking norms from both sides of \eqref{ZEqnNum705143}, the corrupted sensor measurement $y_{i}^{a} (k)$ becomes
	\begin{equation} \label{ZEqnNum862369} 
		\left\| y_{i}^{a} (k)\right\| =\left\| y_{i} (k)+\theta _{i}^{a} (k)1_{p} \right\| . 
	\end{equation} 
	Using the triangular inequality for \eqref{ZEqnNum862369} yields 
	\begin{equation} \label{ZEqnNum171011} 
		\left\| y_{i} (k)\right\| -\left\| \theta _{i}^{a} (k)1_{p} \right\| \le \left\| y_{i}^{a} (k)\right\| \le \left\| y_{i} (k)\right\| +\left\| \theta _{i}^{a} (k)1_{p} \right\| . 
	\end{equation} 
	Based on the bounds of $\theta _{i}^{a} (k)$, given by  \eqref{ZEqnNum165624}, \eqref{ZEqnNum171011}  becomes
	\begin{equation} \label{27)} \nonumber
		\left\| C_{i} \tilde{x}_{i} (k-1)\right\| -\varphi \le \left\| y_{i}^{a} (k)\right\| \le \left\| C_{i} \tilde{x}_{i} (k-1)\right\| +\varphi , 
	\end{equation} 
	which yields
	\begin{equation} \label{ZEqnNum939032} \nonumber
		(\left\| y_{i}^{a} (k)\right\| -\left\| C_{i} \tilde{x}_{i} (k-1)\right\| -\varphi )(\left\| y_{i}^{a} (k)\right\| -\left\| C_{i} \tilde{x}_{i} (k-1)\right\| +\varphi )\le 0. 
	\end{equation} 
	This implies that the condition
	\begin{equation} \label{29)} \nonumber
		\, \left\| y_{i}^{a} (k)-C_{i} \tilde{x}_{i} (k-1)\right\| \le \varphi <\alpha , 
	\end{equation} 
	always holds true. Therefore, under \eqref{ZEqnNum705143}-\eqref{ZEqnNum165624}, the corrupted sensor node $i$ shows non-triggering misbehavior, which proves part 1.
	
	We now prove part 2. Let ${\rm {\mathcal A}}_{n} \subseteq {\rm {\mathcal V}}^{c}$ be the set of sensor nodes showing non-triggering misbehavior. Then, based on the presented result in part 1, under the attack signal  \eqref{ZEqnNum705143}, sensor nodes in the set ${\rm {\mathcal A}}_{n}$ are misled by the attacker and consequently do not transmit any information to their neighbors which make them to act as sink nodes. Since the set of sensor nodes ${\rm {\mathcal A}}_{n} $ is assumed to be a vertex cut. Then, the non-triggering misbehavior of sensor nodes in ${\rm {\mathcal A}}_{n}$ prevents information flow from one portion of the graph ${\rm {\mathcal G}}$ to another portion of the graph ${\rm {\mathcal G}}$ and thus clusters the original graph ${\rm {\mathcal G}}$ into subgraphs. This completes the proof.
\end{proof}

\vspace{-0.3cm}

\noindent
\textbf{Remark 3.} Note that to design the presented strategic false data injection attack signal given in \eqref{ZEqnNum705143} an attacker needs to eavesdrop the actual sensor measurement $y_{i} (k)$ and the last transmitted prior state estimate $\bar{x}_{i} (L)$ through the communication channel. The attacker then determines the   predictive state estimate $\tilde{x}_{i} (k)$ using the dynamics in \eqref{ZEqnNum257073} at each time instant $k\ge L+1$ to achieve non-triggering misbehavior for the sensor node $i$.

We provide Example $1$ for further illustration of the results of Theorem 2. 
\vspace{-0.0cm}
\begin{figure}[!ht]
	\begin{center}
		\includegraphics[width=2.38in,height=2.4in]{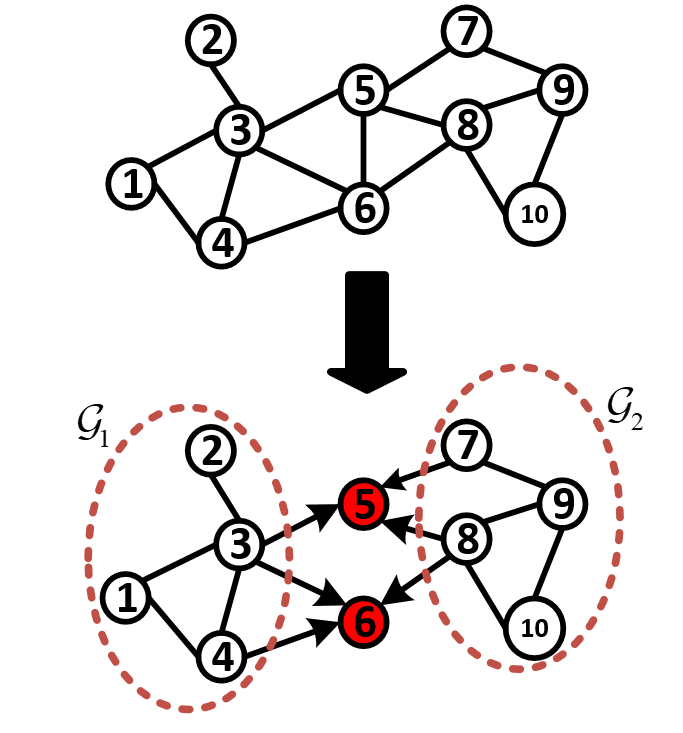}
		\vspace{-2pt}\caption{Effect of non-triggering misbehavior on sensor nodes $\{$5,6$\}$ cluster the graph ${\rm {\mathcal G}}$ in the two isolated graphs ${\rm {\mathcal G}}_{1} $ and ${\rm {\mathcal G}}_{2}.$}\label{fig1}
		\captionsetup{justification=centering}
	\end{center}
\end{figure}
\vspace{-0.15cm}

\noindent
\textbf{Example 1.} Consider a graph topology for a distributed sensor network given in fig. 1. Let the vertex cut ${\rm {\mathcal A}}_{n} =\{ 5,6\}$ be under the presented false data injection attack in Theorem $2$ and show non-triggering misbehavior. Then, the sensor nodes in ${\rm {\mathcal A}}_{n} =\{ 5,6\}$ do not transmit any information to their neighbors under the designed false data injection attack. Moreover, the sensor nodes in ${\rm {\mathcal A}}_{n} =\{ 5,6\}$ act as sink nodes and prevent information flow from subgraph ${\rm {\mathcal G}}_{1}$ to subgraph ${\rm {\mathcal G}}_{2}$ which clusters the graph ${\rm {\mathcal G}}$ into two non-interacting subgraphs ${\rm {\mathcal G}}_{1}$ and ${\rm {\mathcal G}}_{2}$ as shown in Fig. 1. This example shows that the attacker can compromise the vertex cut ${\rm {\mathcal A}}_{n}$ of the original graph ${\rm {\mathcal G}}$ such that it shows non-triggering misbehavior and harm the network connectivity or cluster the graph into various non-interacting subgraphs. 

We now analyze the effect of non-triggering misbehavior on the collective observability of the sensor network. To do so the following definitions are needed.

\smallskip

\noindent \textbf{Definition 5 (Potential Set). } A set of nodes ${\rm {\mathcal P} \subset} {\rm {\mathcal V}}$ is said to be a potential set of the graph ${\rm {\mathcal G}}$ if the pair $(A,C_{{\rm {\mathcal V}}\backslash {\rm{\mathcal P}}} )$ is not collectively observable.

\smallskip

\noindent \textbf{Definition 6 (Minimal Potential Set).} A set of nodes ${\rm {\mathcal P} }_{m} \subset {\rm {\mathcal V}}$ is said to be a minimal potential set if ${\rm {\mathcal P} }_{m}$ is a potential set and no subset of ${\rm {\mathcal P}}_{m}$ is a potential set. 

\smallskip

\noindent \textbf{Remark 4.} Note that if the attacker knows the graph structure and the local pair$(A,C_{i} ),\, \, \, \forall i\in {\mathcal V}$. Then, the attacker can identify the minimum potential set of sensor nodes ${\rm{\mathcal P}}_{m}$ in the graph ${\rm {\mathcal G}}$ and achieves non-triggering misbehavior for ${\rm {\mathcal P} }_{m}.$ Thus, the set of sensor nodes ${\rm {\mathcal P}}_{m}$ does not exchange any information with its neighbors and becomes isolated in the graph ${\rm {\mathcal G}}$.

\smallskip

\noindent \textbf{Corollary 1.}
\textit{Let the set of sensors that shows non-triggering misbehavior be the minimal potential set ${\rm {\mathcal S}}_{n}$. Then, the network is no longer collectively observable and the process state reconstruction from the distributed sensor measurements is impossible.}
\vspace{-0.1cm}
\begin{proof}
	According to the statement of the corollary, ${\rm {\mathcal S}}_{n}$ represents a minimal potential set of the graph ${\rm {\mathcal G}}$ and shows non-triggering misbehavior. Then, the sensor nodes in the set ${\rm {\mathcal S}}_{n}$ do not transmit any information to their neighbors and they act as sink nodes, i.e., they only absorb information. Therefore, the exchange of information happen just between the remaining sensor nodes in the graph ${\rm {\mathcal G}}\backslash {\rm {\mathcal S}}_{n}$.  Hence, after excluding the minimum potential nodes ${\rm {\mathcal S}}_{n}$, the pair $(A,C_{{\rm {\mathcal G}}\backslash {\rm {\mathcal S}}_{n} } )$ becomes unobservable based on the Definitions $5$ and $6$, and thus makes the state reconstruction impossible. This completes the proof. 
\end{proof}

\vspace{-0.4cm}
\subsection{Continuous-triggering Misbehavior}
In this subsection, we discuss how an attacker can compromise the actual sensor measurement to mislead the event-triggered mechanism and achieves continuous-triggering misbehavior and thus results in a time-driven DKF that not only drains the communication resources but also continuously propagates the adverse effect of attack in the network.

\smallskip

\noindent \textbf{Definition 7} \textbf{(Continuous-triggering Misbehavior).} Let the attacker design an attack strategy such that it deceives the triggering mechanism in (\ref{eq3x}) at each time instant. This turns the event-driven DKF into a time-driven DKF that continuously exchanges corrupted information among sensor nodes. We call this a continuous-triggering misbehavior. 

We now show how a reply attack, followed by an eavesdropping attack, can manipulate the sensor reading to cause continuous violation of the event-triggered mechanism (\ref{eq3x}).

\vspace{-0.1cm}

\begin{theorem}
	Consider the process dynamics \eqref{ZEqnNum820040} with $N$ sensor nodes \eqref{ZEqnNum687942} communicating over the graph ${\rm {\mathcal G}}.$ Let the sensor node $i$ in \eqref{ZEqnNum687942} be under a replay attack given by\newline
	\vspace{-0.35cm}
	\begin{equation} \label{ZEqnNum253008} 
		y_{i}^{a} (k)=C_{i} \bar{x}_{i} (k-1)+\upsilon _{i} (k),\, \, \forall k\ge l+1, 
	\end{equation} 
	
	\vspace{-0.15cm}
	\noindent
	where $\bar{x}_{i}(k-1)$ represents the last transmitted prior state and $\upsilon_{i} (k)$ denotes a scalar disruption signal. $l$ denotes the last triggering time instant when intact prior state estimate was transmitted. Then, the sensor node $i$ shows continuous-triggering misbehavior if the attacker selects $\left\| \upsilon _{i} (k)\right\| >\alpha.$ 
\end{theorem}{}

\begin{proof}
	To mislead a sensor to cause a continuous-triggering misbehavior, the attacker needs to design the attack signal such that the event-triggered condition (\ref{eq3x})    is constantly being violated, i.e., $\, \left\| y_{i}^{a} (k)-C_{i} \tilde{x}_{i} (k-1)\right\| \ge \alpha $ all the time. The attacker can eavesdrop the last transmitted prior state estimate $\bar{x}_{i}(k-1)$ and design the strategic attack signal given by \eqref{ZEqnNum253008}. Then, one has
	\vspace{-0.15cm}
	\begin{equation} \label{ZEqnNum491202} 
		\begin{array}{l} {y_{i}^{a} (k)-C_{i} \tilde{x}_{i} (k-1)=C_{i} \bar{x}_{i} (k-1)+\delta _{i} (k)-C_{i} \tilde{x}_{i} (k-1)} \\ 
			\, \, \, \, \, \, \, \, \, \, \, \, \, \, \, \, \, \, \, \, \, \, \, \,  \,\, \, \, \, \, \, \, \, \, \, \, \,{=C_{i} \bar{x}_{i} (k-1)+\upsilon _{i} (k)}   -C_{i} [\zeta _{i} (k-1)\bar{x}_{i} (k-1)\\
			\, \, \, \, \, \, \, \, \, \, \, \, \, \, \, \, \, \, \, \, \, \, \, \,  \,\, \, \, \, \, \, \, \, \, \, \, \,+(1-\zeta _{i} (k-1))A\bar{x}_{i} (k-2)] \\
			 \, \, \, \, \, \, \, \, \, \, \, \, \, \, \, \, \, \, \, \, \, \, \, \,  \,\, \, \, \, \, \, \, \, \, \, \, \,{=(1-\zeta _{i} (k-1))C_{i} [\bar{x}_{i} (k-1)-A\bar{x}_{i} (k-2)]+\upsilon _{i} (k),} \end{array} 
	\end{equation} 
	Taking the norm from both sides of \eqref{ZEqnNum491202} yields
	\begin{equation} \label{ZEqnNum734745} 
		\begin{array}{l} {\left\| y_{i}^{a} (k)-C_{i} \tilde{x}_{i} (k-1)\right\| } \\ {=\left\| (1-\zeta _{i} (k-1))C_{i} [\bar{x}_{i} (k-1)-A\bar{x}_{i} (k-2)]+\upsilon _{i} (k)\right\| ,} \end{array} 
	\end{equation} 
	Since for $k=l+1$, $\zeta _{i}(l)=1$
	\vspace{-0.15cm}
	\begin{equation} \label{34)} 
		\left\| y_{i}^{a} (l+1)-C_{i} \tilde{x}_{i} (l)\right\| =\left\| \upsilon _{i} (l+1)\right\| , 
	\end{equation} 
	{If the attacker selects $\upsilon _{i}(l+1)$ in \eqref{34)}} such that $\left\| \upsilon _{i} (l+1)\right\| >\alpha $, then  the attack signal \eqref{ZEqnNum253008} ensures triggering at time instant $k=l+1.$ Then, based on similar argument for \eqref{ZEqnNum734745}, $\forall k\ge l+1$
	\vspace{-0.15cm}
	\begin{equation} \label{35)} \nonumber
		\left\| y_{i}^{a} (k)-C_{i} \tilde{x}_{i} (k-1)\right\| =\left\| \upsilon _{i} (k)\right\| >\alpha , 
	\end{equation} 
	which ensures continuous triggering misbehavior. This completes the proof. 
\end{proof}                             
\vspace{-0.25cm}

To achieve continuous-triggering misbehavior the attacker needs to eavesdrop prior state estimates $\bar{x}_{i} (k-1)$ at each triggering instant and selects the $\upsilon _{i}(k)$ large enough such that $\left\| \upsilon _{i} (k)\right\| >\alpha $ always holds true. 

Note that continuous-triggering misbehavior can completely ruin the advantage of event-triggered mechanisms and turn it into time-driven mechanisms. This significantly increases the communication burden. Since nodes in the WSNs are usually powered through batteries with limited energy, the attacker can drain sensors limited resources by designing the above-discussed attack signals to achieve continuous-triggering misbehavior, and, consequently can make them non-operating in the network along with the deteriorated performance of the network.

Note that although we classified attacks into non-triggering misbehavior and continuous-triggering misbehavior, to analyze how the attacker can leverage the event-triggered mechanism, the following \textit{analysis, detection and mitigation approaches} are not restricted to any class of attacks. 
\vspace{-0.3cm}
\section{ Attack Detection}
In this section, we present an entropy estimation-based attack detection approach for the event-triggered DKF. 

The KL divergence is a non-negative measure of the relative entropy between two probability distributions which is defined as follows.

\noindent \textbf{Definition 8 (KL Divergence) \cite{c24}}. Let $X$ and $Z$ be two random variables with probability density function $P_{X}$ and $P_{Z}$, respectively. The KL divergence measure between $P_{X}$ and $P_{Z}$ is defined as
\vspace{-0.2cm}
\begin{equation} \label{ZEqnNum937457} 
	D_{KL} (P_{X} ||P_{Z} )=\int _{\theta \in \Theta }P_{X} (\theta )\log \left(\frac{P_{X} (\theta )}{P_{Z} (\theta )} \right) , 
\end{equation} 
with the following properties \cite{c32}
\begin{enumerate}
	\item  $D_{KL} (P_{X} ||P_{Z} )\ge 0;$
	\item  $D_{KL} (P_{X} ||P_{Z} )=0$  if and only if,  $P_{X} =P_{z} ;$ 
	\item  $D_{KL} (P_{X} ||P_{Z} )\ne D_{KL} (P_{Z} ||P_{X} ).$
\end{enumerate}

In the existing resilient literature, the entropy-based anomaly detectors need to know the probability density function of sequences, i.e., $P_{X}$ and $P_{Z},$ {in \eqref{ZEqnNum937457}} to determine the relative entropy. In most of the cases, authors assume that the probability density function of corrupted innovation sequence remains Gaussian (see \cite{c24} and \cite{c34} for instance). Since, the attacker's input signal is unknown, it is restrictive to assume that the probability density function of the corrupted sequence remains Gaussian. To relax this \textit{restrictive assumption} on probability density function of the corrupted sequence, we estimate the relative entropy between two random sequences $X$ and $Z$ using \textit{$k-$nearest neighbor $(k-NN)$} based divergence estimator \cite{d5}.

{Let $\{ X_{1},\ldots,X_{n_{1} } \} $ and $\{ Z_{1} ,\ldots ,Z_{n_{2} } \} $ be i.i.d. samples drawn independently from $P_{X} $ and $P_{Z},$ respectively with $X_{j},\,\, Z_{j} \in {\bf {\mathbb R}}^{m}$. Let $d_{k}^{X}(i)$ be the Euclidean distance between $X_{i}$ and its \textit{$k-NN$} in $\{ X_{l} \} _{l\ne i} .$ The \textit{$k-NN$} of a sample $s$ in $\{ s_{1} ,\ldots ,s_{n} \} $ is $s_{i(k)}$ where $i(1),\ldots,i(n)$ such that
	\vspace{-0.2cm}
	\begin{equation} \label{37)} \nonumber
		\left\| s-s_{i(1)} \right\| \le \left\| s-s_{i(2)} \right\| \le \ldots \le \left\| s-s_{i(n)} \right\| . 
	\end{equation} 
	More specifically, the Euclidean distance $d_{k}^{X}(i)$ is given by \cite{d5a}
	\begin{equation}
		\nonumber
		\begin{array}{l}
			d_k^X(i) = \mathop {\min }\limits_{j = 1, \ldots ,{n_1},j \ne \{ i,{j_1},...,{j_{k - 1}}\} } \left\| {{X_i} - {X_j}} \right\|
		\end{array}
\end{equation}}
The \textit{$k-NN$} based relative entropy estimator is given by \cite{d5}
\begin{equation} \label{ZEqnNum207466} 
	\hat{D}_{KL} (P_{X} ||P_{Z} )=\frac{m}{n_{1} } \sum _{i=1}^{n_{1} }\log \frac{d_{k}^{Z} (i)}{d_{k}^{X} (i)}  +\log \frac{n_{2} }{n_{1} -1} . 
\end{equation} 

The innovation sequences represent the deviation of the actual output of the system from the estimated one. It is known that innovation sequences approach a steady state quickly and thus it is reasonable to design innovation-based anomaly detectors to capture the system abnormality \cite{c24}. Using the innovation sequence of each sensor and the innovation sequences that it estimates for its neighbors, we present innovation based divergence estimator and design detectors to capture the effect of the attacks on the event-triggered DKF. 

Based on innovation expression \eqref{ZEqnNum276515}, in the presence of attack, one can write the compromised innovation $r_{i}^{a} (k)$ for sensor node $i$ with disrupted measurement $y_{i}^{a} (k)$ in \eqref{ZEqnNum973066} and state estimation $\bar{x}_{i}^{a} \, $ based on \eqref{ZEqnNum120276} as
\vspace{-0.15cm}
\begin{equation} \label{ZEqnNum255093} 
	r_{i}^{a} (k)=y_{i}^{a} (k)-C_{i} \bar{x}_{i}^{a} (k). 
\end{equation} 

Let $\{ r_{i}^{a} (l),\ldots ,r_{i}^{a} (l-1+w)\} $ and $\{ r_{i} (l),\ldots ,r_{i} (l-1+w)\} $ be i.i.d. \textit{p}-dimensional samples of corrupted and nominal innovation sequences with probability density function $P_{r_{i}^{a} } $  and $P_{r_{i} },$ respectively. The nominal innovation sequence follows $r_{i}(k)$ defined in \eqref{ZEqnNum276515}.  Using \textit{$k-NN$} based relative entropy estimator \eqref{ZEqnNum207466}, one has \cite{d5}
\begin{equation} \label{ZEqnNum280433} 
	\hat{D}_{KL} (P_{r_{i}^{a} } ||P_{r_{i} } )=\frac{p}{w} \sum _{j=1}^{w}\log \frac{d_{k}^{r_{i} } (j)}{d_{k}^{r_{i}^{a} } (j)}  +\log \frac{w}{w-1} ,\, \, \forall i\in {\rm {\mathcal V}}. 
\end{equation} 

Define the average of the estimated KL divergence over a time window of $T$ as
\begin{equation} \label{ZEqnNum738078} 
	\Phi _{i} (k)=\frac{1}{T} \sum _{l=k-T+1}^{k}\hat{D}_{KL} (P_{r_{i}^{a} } ||P_{r_{i} } ) ,\, \, \forall i\in {\rm {\mathcal V}}. 
\end{equation} 

Now, in the following theorem, it is shown that the effect of attacks on the sensors can be captured using \eqref{ZEqnNum738078}.

\begin{theorem}
	Consider the distributed sensor network \eqref{ZEqnNum820040}-\eqref{ZEqnNum687942} under attack on sensor. Then, 
	\begin{enumerate}
		\item  in the absence of attack, $\Phi _{i} (k)=\log (w/w-1),\, \, \, \forall k;$ 
		\item  in the presence of attack, $\Phi _{i} (k)>\delta ,\, \, \forall k>l_{a},$ where $\delta $ and $l_{a}$ denotes, respectively, a predefined threshold and the time instant at which the attack happen. 
	\end{enumerate}
\end{theorem}{}

\begin{proof}
	In the absence of attack, the samples of innovation sequences $\{ r_{i}^{a} (l),\ldots ,r_{i}^{a} (l-1+w)\} $ and $\{ r_{i} (l),\ldots ,r_{i} (l-1+w)\} $ are similar. Then, the Euclidean distance $d_{k}^{r_{i}^{a} } (j)=d_{k}^{r_{i} } (j),\, \, \forall j\in \{ 1,...,w\} $ and one has
	\begin{equation} \label{ZEqnNum663932} 
		\hat{D}_{KL} (P_{r_{i}^{a} } ||P_{r_{i} } )=\log \frac{w}{w-1} ,\, \, \forall i\in {\rm {\mathcal V}}. 
	\end{equation} 
	Based on \eqref{ZEqnNum663932}, one has 
	\vspace{-0.15cm}
	\begin{equation} \label{a43)} 
		\Phi _{i} (k)=\frac{1}{T} \sum _{l=k-T+1}^{k}\log \frac{w}{w-1}  =\log \frac{w}{w-1} < \delta ,\, \, \forall i\in {\rm {\mathcal V}}. 
	\end{equation} 
	{where $\log (w/w-1)$ in (42) depends on the sample size of innovation sequence and $\log (w/w-1)\le 0.1,\, \, \, \forall w\ge 10$. Therefore, the predefined threshold $\delta$ can be selected with some $\delta>0.1$  such that the condition in (42) is always satisfied.} This complete the proof of part 1.
	
	In the presence of attack, the samples of innovation sequences $\{ r_{i}^{a} (l),\ldots ,r_{i}^{a} (l-1+w)\} $ and $\{ r_{i} (l),\ldots ,r_{i} (l-1+w)\} $ are different, i.e., $d_{k}^{r_{i}^{a} } (j)\ne d_{k}^{r_{i} } (j),\, \, \forall j\in \{ 1,...,w\} $. More specifically, $d_{k}^{r_{i} } (j)>d_{k}^{r_{i}^{a} } (j), \, \, \forall j\in \{ 1,...,w\} $ due to change in the corrupted innovation sequence. Therefore, based on \eqref{ZEqnNum280433} the estimated relative entropy between sequences becomes
	\begin{equation} \label{ZEqnNum657988} 
		\hat{D}_{KL} (P_{r_{i}^{a} } ||P_{r_{i} } )=\frac{p}{w} \sum _{j=1}^{w}\log (1+\frac{\Delta _{k}^{r_{i} } (j)}{d_{k}^{r_{i}^{a} } (j)} ) +\log \frac{w}{w-1} ,\, \forall i\in {\rm {\mathcal V}}, 
	\end{equation} 
	with $\Delta _{k}^{r_{i} } (j)$ as the change in Euclidean distance due to corrupted innovation sequence. Based on \eqref{ZEqnNum657988}, one has 
	\begin{equation} \label{ZEqnNum750552} 
		\hat{D}_{KL} (P_{r_{i}^{a} } ||P_{r_{i} } )=\frac{p}{w} \sum _{j=1}^{w}\log (1+\frac{\Delta _{k}^{r_{i} } (j)}{d_{k}^{r_{i}^{a} } (j)} ) +\log \frac{w}{w-1} \gg \log \frac{w}{w-1} . 
	\end{equation} 
	Thus, one has
	\vspace{-0.2cm}
	\begin{equation} \label{46)} 
		\Phi _{i} (k)=\frac{1}{T} \sum _{l=k-T+1}^{k}\hat{D}_{KL} (P_{r_{i}^{a} } ||P_{r_{i} } )>\delta  ,\, \, \forall i\in {\rm {\mathcal V}}, 
	\end{equation} 
	where $T$ and $\delta $ denote the sliding window size and the predefined design threshold. This completes the proof.
\end{proof}
Based on Theorem 4, one can use the following condition for attack detection.
\vspace{-0.2cm}
\begin{equation} \label{ZEqnNum861796} 
	\left\{\begin{array}{l} {\Phi _{i} (k)\, <\delta :H_{0}, } \\ {\Phi _{i} (k)\, >\delta \, \, :H_{1}, } \end{array}\right. 
\end{equation} 
where $\delta $ denotes the designed threshold for detection, the null hypothesis $H_{0} $ represents the intact mode of sensor nodes and $H_{1}$ denotes the compromised mode of sensor nodes.

\smallskip
\noindent \textbf{Remark 5.} {Note that in the absence of an attack, the innovation sequence has a known zero-mean Gaussian distribution due to the measurement noise. Based on the prior system knowledge, one can always consider that the nominal innovation sequence is zero-mean Gaussian distribution with predefined covariance. The bound on the predefined covariance can be determined during normal operation of the event-triggered DKF.  This assumption for the knowledge of the nominal innovation sequence for attack detection is standard in the existing literature (see \cite{c34} for instance). The designed threshold $\delta $ in \eqref{ZEqnNum861796} is a predefined parameter and chosen appropriately for the detection of the attack signal. Moreover, the selection of detection threshold based on expert knowledge is standard in the existing literature. For example, several results on adversary detection and stealthiness have considered similar thresholds \cite{c24}-\cite{c26}. } 

\begin{algorithm}[!ht]
	\caption{Detecting attacks on sensors.}
		\begin{enumerate}
			\item [1:] Initialize with a time window $T$ and detection threshold $\delta.$
			\item [2:] \textbf{procedure} $\forall i=1,\ldots ,N$ 
			\item [3:] Use samples of innovation sequences $\{ r_{i}^{a} (l),\ldots,$ \qquad $r_{i}^{a} (l-1+w)\} $ and $\{ r_{i} (l),\ldots ,r_{i} (l-1+w)\} $ based on \eqref{ZEqnNum255093} and \eqref{ZEqnNum276515}, $\forall l\in k.$
			\item [4:] Estimate the $\hat{D}_{KL} (P_{r_{i}^{a} } ||P_{r_{i} } )$ using \eqref{ZEqnNum750552}.
			\item [5:] Compute $\Phi _{i} (k)$ as \eqref{46)}  and use condition in \eqref{ZEqnNum861796} to detect attacks on sensors.
			\item [6:] \textbf{end procedure}
		\end{enumerate}
\end{algorithm}
Based on the results presented in Theorem 4 and Algorithm 1, one can capture attacks on both sensors and communication links, but it cannot identify the specific compromised communication link as modelled in \eqref{ZEqnNum397788}. To detect the source of attacks, we present an estimated entropy-based detector to capture the effect of attacks on the specific communication channel. More specifically, the relative entropy between the estimated innovation sequences for the neighbors at particular sensor node and the nominal innovation sequence of the considered sensor node is estimated using \eqref{ZEqnNum207466}.

Define the estimated innovation sequences $\zeta _{i,j}^{a}(k)$ for a neighbor $j$ under attacks on communication channel from the sensor node $i$ side as 
\vspace{-0.15cm}
\begin{equation} \label{ZEqnNum178443} 
	\zeta _{i,j}^{a} (k)=y_{i} (k)-C_{j} \tilde{x}_{j}^{a} (k), 
\end{equation} 
where $\tilde{x}_{j}^{a}(k)$ is the corrupted communicated state estimation of neighbor $j$ at sensor node $i$ at the last triggering instant.\par 
Let $\{ \zeta _{i,j}^{a} (l),\ldots ,\zeta _{i,j}^{a} (l-1+w)\}$ be i.i.d. \textit{p}-dimensional samples of neighbor's estimated innovation at the sensor node $i$
with probability density function $P_{\zeta _{i,j}^{a} }.$  Using \textit{$k-NN$} based relative entropy estimator \eqref{ZEqnNum207466}, one has
\begin{equation} \label{ZEqnNum691139} 
	\hat{D}_{KL} (P_{\zeta _{i,j}^{a} } ||P_{r_{i} } )=\frac{p}{w} \sum _{j=1}^{w}\log \frac{d_{k}^{r_{i} } (j)}{d_{k}^{\zeta _{i,j}^{a} } (j)}  +\log \frac{w}{w-1} ,\, \, \forall i\in {\rm {\mathcal V}},j\in N_{i} . 
\end{equation} 

Note that in the presence of attacks on the communication channels, the neighbor's actual innovation differs the neighbor's estimated innovation at sensor $i$. {In the absence of the attack, the mean value of all the sensor state estimates converge to the mean of the desired process state at steady state, and, therefore, the innovation sequences $r_{i}$ and $\zeta _{i,j}^{a}$ have the same zero mean Gaussian distributions. In the presence of attack, however, as shown in Theorem 5 and Algorithm 2, their distributions diverge.}

Define the average of the KL divergence over a time window of $T$ as
\vspace{-0.2cm}
\begin{equation} \label{ZEqnNum932962} 
	\Psi _{i,j} (k)=\frac{1}{T} \sum _{l=k-T+1}^{k}\hat{D}_{KL} (P_{\zeta _{i,j}^{a} } ||P_{r_{i} } ) ,\, \, \forall i\in {\rm {\mathcal V}},\, j\in N_{i} . 
\end{equation} 

\begin{theorem}
	Consider the distributed sensor network \eqref{ZEqnNum820040}-\eqref{ZEqnNum687942} under attack on communication links \eqref{ZEqnNum397788}. Then, in the presence of an attack, $\Psi _{i,j} (k)>\delta ,\, \, \forall k$ where $\delta $ denotes a predefined threshold.
\end{theorem}

\begin{proof}
	The result follows a similar argument as given in the proof of part $2$ of Theorem 4.  
\end{proof}
\begin{algorithm}[!ht]
	\caption{Detecting attack on a specific communication link.}
		\begin{enumerate}
			\item [1:] Initialize with a time window $T$ and detection threshold $\delta.$
			\item [2:] \textbf{procedure} $\forall i=1,\ldots ,N$ 
			\item [3:] For each sensor node $j\in N_{i} $, use samples of innovation sequences$\{ \zeta _{i,j}^{a} (l),\ldots ,\zeta _{i,j}^{a} (l-1+w)\} $ and  $\{ r_{i} (l),\ldots ,r_{i} (l-1+w)\}$ based on \eqref{ZEqnNum178443} and \eqref{ZEqnNum276515}, $\forall l\in k.$
			\item [4:] Estimate the $\hat{D}_{KL} (P_{\zeta _{i,j}^{a} } ||P_{r_{i} } )$ using \eqref{ZEqnNum691139}.
			\item [5:] Compute $\Psi _{i,j}(k)$ as \eqref{ZEqnNum932962} and use same argument in \eqref{ZEqnNum861796} to detect attacks on specific communication link.
			\item [6:] \textbf{end procedure}
		\end{enumerate}
\end{algorithm}

\vspace{-0.4cm}
\section{ Secure Distributed Estimation Mechanism}

This section presents a meta-Bayesian approach for secure event-triggered DKF, which incorporates the outcome of the attack detection mechanism to perform second-order inference and consequently form beliefs over beliefs. That is, the second-order inference forms confidence and trust about the truthfulness or legitimacy of the sensors' own state estimate (i.e., the posterior belief of the first-order Bayesian inference) and those of its neighbor's state estimates, respectively. Each sensor communicates its confidence to its neighbors. Then sensors incorporate the confidence of their neighbors and their own trust about their neighbors into their posterior update laws to successfully discard the corrupted information. 

\vspace{-0.4cm}

\noindent 
\subsection{Confidence of sensor nodes}
The second-order inference forms a confidence value for each sensor node which determines the level of trustworthiness of the sensor about its own measurement and state estimate (i.e., the posterior belief of the first-order Bayesian inference).  If a sensor node is compromised, then the presented attack detector detects the adversary and it then reduces its level of trustworthiness about its own understanding of the environment and communicates it with its neighbors to inform them the significance of its outgoing information and thus slow down the attack propagation. 

To determine the confidence of the sensor node $i$, based on the divergence $\hat{D}_{KL} (P_{r_{i}^{a} } ||P_{r_{i} } )$ from Theorem 4, we first define
\vspace{-0.15cm}
\begin{equation} \label{ZEqnNum125869} 
	\chi _{i} (k)=\frac{\Upsilon _{1} }{\Upsilon _{1} +\hat{D}_{KL} (P_{r_{i}^{a} } ||P_{r_{i} } )} , 
\end{equation} 
with $0<\Upsilon _{1} <1$ represents a predefined threshold to account for the channel fading and other uncertainties. Then, in the following lemma, we formally present the results for the confidence of sensor node $i$.

\noindent \textbf{Lemma 1.} \textit{Let $\beta _{i} (k)$ be the confidence of the sensor node $i$ which is updated using 
	\begin{equation} \label{ZEqnNum359584} 
		\beta _{i} (k)=\sum _{l=0}^{k-1}(\kappa _{1} )^{k-l+1}  \chi _{i} (l), 
	\end{equation} 
	where $\chi _{i}(k)$ is defined in \eqref{ZEqnNum125869}, and $0<\kappa _{1}<1$ is a discount factor. Then, $\beta _{i}(k)\in (0,1]$ and 
	\begin{enumerate}
		\item  $\beta _{i} (k)\to 0,\, \, \, \forall i\in {\rm {\mathcal V}}^{c} ;$
		\item  $\beta _{i} (k)\to 1,\, \, \, \forall i\in {\rm {\mathcal V}}\backslash {\rm {\mathcal V}}^{c} .$
\end{enumerate}}

\begin{proof}
	Based on the expression \eqref{ZEqnNum125869}, since $\hat{D}_{KL} (P_{r_{i}^{a} } ||P_{r_{i} } )\ge 0$, one has $\chi _{i} (k)\in (0,1]$. Then, using \eqref{ZEqnNum359584}, one can infer that $\beta _{i} (k)\in (0,1]$. 
	
	Now according to Theorem 4, if the sensor node $i$ is under attack, then $\hat{D}_{KL} (P_{r_{i}^{a} } ||P_{r_{i} } )>>\Upsilon _{1} $ in \eqref{ZEqnNum125869}, which makes $\chi _{i}(k)$ close to zero. Then, based on expression \eqref{ZEqnNum359584} with the discount factor $0<\kappa _{1} <1,$ the confidence $\beta _{i}(k)$ in \eqref{ZEqnNum359584} approaches zero, and thus the $i^{th} $ sensor's belief about the trustworthiness of its own information would be low. This completes the proof of part 1. 
	
	On the other hand, based on Theorem 4, in the absence of attacks, $\hat{D}_{KL} (P_{r_{i}^{a} } ||P_{r_{i} } )\to 0$ as $w\to \infty $, which makes $\chi _{i} (k)$ close to one and, consequently, $\beta _{i} (k)$ becomes close to one. This indicates that the $i^{th}$ sensor node is confident about its own state estimate. This completes the proof of part 2.
\end{proof}

\vspace{-0.1cm}
Note that the expression for the confidence of sensor node $i$ in \eqref{ZEqnNum359584} can be implemented using the following difference equation
\vspace{-0.3cm}
\begin{equation} \label{53)} \nonumber
	\beta _{i} (k+1)=\beta _{i} (k)+\kappa _{1} \chi _{i} (k). 
\end{equation} 

Note also that the discount factor in \eqref{ZEqnNum359584} determines how much we value the current experience with regards to past experiences. It also guarantees that if the attack is not persistent and disappears after a while, or if a short-period adversary rather than attack (such as packet dropout) causes, the belief will be recovered, as it mainly depends on the current circumstances.

\vspace{-0.35cm}

\noindent 
\subsection{Trust  of sensor nodes about their incoming information}

Similar to the previous subsection, the second-order inference forms trust of sensor nodes to represent their level of trust on their neighboring sensor's state estimates. Trust decides the usefulness of the neighboring information in the state estimation of sensor node $i$. 

The trust of the sensor node $i$ on its neighboring sensor $j$ can be determined based on the divergence $\hat{D}_{KL} (P_{\zeta _{i,j}^{a} } ||P_{r_{i} })$ in \eqref{ZEqnNum178443} from Theorem 5, from which we define 
\begin{equation} \label{ZEqnNum846884} 
	\theta _{i,j} (k)=\frac{\Lambda _{1} }{\Lambda _{1} +\hat{D}_{KL} (P_{\zeta _{i,j}^{a} } ||P_{r_{i} } )} , 
\end{equation} 
where $0<\Lambda _{1} <1$ represents a predefined threshold to account for the channel fading and other uncertainties. Then, in the following lemma, we formally present the results for the trust of the sensor node $i$ on its neighboring sensor $j.$

\smallskip

\noindent \textbf{Lemma 2.} \textit{Let $\sigma _{i,j}(k)$ be the trust of the sensor node  $i$ on its neighboring sensor $j$ which is updated using 
	\begin{equation} \label{ZEqnNum805360} 
		\sigma _{i,j} (k)=\sum _{l=0}^{k-1}(\kappa _{2} )^{k-l+1}  \theta _{i,j} (l), 
	\end{equation} 
	where $\theta _{i,j}(k)$ is defined in \eqref{ZEqnNum846884}, and $0<\kappa _{2} <1$ is a discount factor. Then, $\sigma _{i,j}(k)\in (0,1]$ and 
	\begin{enumerate}
		\item  $\sigma _{i,j} (k)\to 0,\, \, \, \forall j\in {\rm {\mathcal V}}^{c} \cap N_{i} ;$
		\item  $\sigma _{i,j} (k)\to 1,\, \, \, \forall j\in {\rm {\mathcal V}}\backslash {\rm {\mathcal V}}^{c} \cap N_{i} .$
\end{enumerate}}

\begin{proof}
	The result follows a similar argument as given in the proof of Lemma 1.
\end{proof}
\vspace{-0.2cm}
Note that the trust of sensor node $i$ in \eqref{ZEqnNum805360} can be implemented using the following difference equation 
\vspace{-0.2cm}
\begin{equation} \label{56)} \nonumber
	\sigma _{i,j} (k+1)=\sigma _{i,j} (k)+\kappa _{2} \theta _{i,j} (k). 
\end{equation} 
Using the presented idea of trust, one can identify the attacks on the communication channel and discard the contribution of compromised information for the state estimation.

\vspace{-0.35cm}
\subsection{Attack mitigation mechanism using confidence and trust of sensors}
This subsection incorporates the confidence and trust of sensors to design a resilient event-triggered DKF. To this end, using the presented confidence $\beta _{i}(k)$ in \eqref{ZEqnNum359584} and trust $\sigma _{i,j}(k)$ in \eqref{ZEqnNum805360}, we design the resilient form of the event-triggered DKF as 
\begin{equation} \label{ZEqnNum565391} 
	\begin{array}{l} {\hat{x}_{i} (k)=\bar{x}_{i} (k)+K_{i} (k)(\beta _{i} (k)y_{i} (k)+(1-\beta _{i} (k))C_{i} m_{i} (k)-C_{i} \bar{x}_{i} (k))} \\ {\, \, \, \, \, \,\, \, \, \, \,\, \,\, \,\, \, \, \, \, \,\, \, \, \, +\gamma _{i} \sum _{j\in N_{i} }\sigma _{i,j} (k)\beta _{j} (k)(\tilde{x}_{j} (k)-\tilde{x}_{i}  (k)),} \end{array} 
\end{equation} 
where the weighted neighbor's state estimate $m_{i}(k)$ is defined as
\vspace{-0.2cm}
\begin{equation} \label{ZEqnNum466700} 
	\begin{array}{l} {m_{i} (k)=\frac{1}{\left|N_{i} \right|} \sum _{j\in N_{i} }\sigma _{i,j} (k)\beta _{j} (k)\tilde{x}_{j} (k) \approx x(k)+\varepsilon _{i} (k),\, \, \, } \\ {\, \, \, \, \, \, \, \, \, \, \, \, \, \, \qquad \forall k\, \, \, \left\| \varepsilon _{i} (k)\right\| <\tau ,} \end{array} 
\end{equation} 
where $\varepsilon _{i}(k)$ denotes the deviation between the weighted neighbor's state estimate $m_{i} (k)$ and the actual process state $x(k)$. Note that in \eqref{ZEqnNum466700} the weighted state estimate depends on the trust values $\sigma _{i,j} (k)$ and the confidence values $\beta _{j} (k),\, \, \forall j\in N_{i}.$ Since the weighted state estimate depends only on the information from intact neighbors, then one has $\left\| \varepsilon _{i} (k)\right\| <\tau$ for some $\tau >0,\, \, \forall k.$ For the sake of mathematical representation, we approximate the weighted state estimate $m_{i}(k)$ in terms of the actual process state $x(k)$, i.e., $m_{i}(k)\approx x(k)+\varepsilon _{i} (k).$ We call this a meta-Bayesian inference that integrates the first-order inference (state estimates) with second-order estimates or belief (trust and confidence on the trustworthiness of state estimate beliefs). 

Define the prior and predictive state estimation errors as
\begin{equation} \label{ZEqnNum250987} 
	\begin{array}{l} {\bar{\eta }_{i} (k)=x(k)-\bar{x}_{i} (k)} \\ {\tilde{\eta }_{i} (k)=x(k)-\tilde{x}_{i} (k),} \end{array} 
\end{equation} 
Using the threshold in triggering mechanism (\ref{eq3x}), one has 
\begin{equation} \label{ZEqnNum528573} 
	\begin{array}{l} {\left\| \tilde{\eta }_{i} (k)\right\| -\left\| x(k+1)-x(k)+v_{i} (k+1)\right\| \le \alpha /\left\| C_{i} \right\| ,} \\ {\left\| \tilde{\eta }_{i} (k)\right\| \le \alpha /\left\| C_{i} \right\| +{\rm {\mathcal B}},} \end{array} 
\end{equation} 
where ${\rm {\mathcal B}}$ denotes the bound on $\left\| x(k+1)-x(k)+v_{i} (k+1)\right\| .$ 

\noindent Other notations used in the following theorem are given by 
\begin{equation} \label{ZEqnNum500695} 
	\begin{array}{l} {\bar{\eta }(k)=[\bar{\eta }_{1} (k),\ldots ,\bar{\eta }_{N} (k)],\, \, \, M(k)=diag[M_{1} (k),\ldots ,M_{N} (k)]} \\ {\Upsilon =diag[\gamma _{1} ,\ldots ,\gamma _{N} ],\, \, \Upsilon _{m} =\left\| \max \{ \gamma _{i} \} \right\| ,\, \, \forall i \in \mathcal{V}}, \\ {\bar{\beta }=(I_{N} -diag(\beta _{i} )),\, \, \, \, E(k)=[\varepsilon _{1} (k),\ldots ,\varepsilon _{N} (k)],} \\ {\tilde{\eta }(k)=[\tilde{\eta }_{1} (k),\ldots ,\tilde{\eta }_{N} (k)].} \end{array} 
\end{equation} 

\noindent
\textbf{Assumption 4.} At least $({\rm {\mathcal C}}(N_{i} )/2)+1$ neighbors of the sensor node $i$ are intact. 

Assumption 4 is similar to the assumption found in the secure estimation and control literature \cite{c19}, \cite{c29}. Necessary and sufficient condition for any centralized or distributed estimator to resiliently estimate actual state is that the number of attacked sensors is less than half of all sensors.

\smallskip
\noindent
\noindent \textbf{Remark 6.} {Note that the proposed notion of trust and confidence for hybrid attacks on sensor networks for event-triggered DKF can also be seen as the weightage in the covariance fusion approach. Although covariance intersection-based Kalman consensus filters have been widely used in the literature to deal with unknown correlations in sensor networks (for instants see \cite{c10}-\cite{d10} and \cite{c310}-\cite{c312}), most of these results considered the time-triggered distributed state estimation problem with or without any adversaries.  Compared with the existing results, however, a novelty of this work lies in detecting and mitigating the effect of attacks on sensors and communication channels for event-triggered DKF and providing a rigorous mathematical analysis for different triggering misbehaviors.}

\begin{theorem}
	Consider the resilient event triggered DKF \eqref{ZEqnNum565391} with the triggering mechanism (\ref{eq3x}). Let the time-varying graph be ${\rm {\mathcal G}}(k)$ such that at each time instant $k,$ Assumptions 3 and 4 are satisfied. Then,  
	\begin{enumerate}
		\item The following uniform bound holds on state estimation error in \eqref{ZEqnNum250987}, despite attacks
		\vspace{-0.2cm}
		\begin{equation} \label{ZEqnNum232225} 
			\left\| \bar{\eta }(k)\right\| \le (A_{o} )^{k} \left\| \bar{\eta }(0)\right\| +\sum _{m=0}^{k-1}(A_{o} )^{k-m-1}  B_{o} , 
		\end{equation} 
		where
		\vspace{-0.2cm}
		\begin{equation} \label{ZEqnNum594295} 
			\begin{array}{l} {A_{o} =\sigma _{\max } ((I_{N} \otimes A)M(k)),} \\ {B_{o} =\sigma _{\max } (A)\sigma _{\max } ({\rm {\mathcal L}}(k))\Upsilon _{m} \sqrt{N(} \alpha /\left\| C_{i} \right\| +{\rm {\mathcal B}})} \\ {\, \, \, \, \,\, \, \, \, \, \, \,\, \, \, \, \, \, \, \, +(\sigma _{\max } (A)+\sigma _{\max } (A_{o} ))\left\| \bar{\beta }\right\| \sqrt{N} \tau,} \end{array} 
		\end{equation} 
		with ${\rm {\mathcal L}}(k)$ denotes the confidence and trust dependent time-varying graph Laplacian matrix, and bound $\tau $ defined in \eqref{ZEqnNum466700};
		\item The uniform bound on the state estimation error \eqref{ZEqnNum232225} becomes
		\vspace{-0.15cm}
		\begin{equation} \label{64)} 
			{\mathop{\lim }\limits_{k\to \infty }} \left\| \bar{\eta }(k)\right\| \le \frac{A_{o} B_{o} }{1-A_{o} }.
		\end{equation} 
	\end{enumerate}
	Moreover, other notations used in \eqref{ZEqnNum594295} are defined in \eqref{ZEqnNum500695}. 
\end{theorem}

\begin{proof}
	Using the presented resilient estimator \eqref{ZEqnNum565391}, one has
	\begin{equation} \label{ZEqnNum429555} 
		\begin{array}{l} {\bar{x}_{i} (k+1)=A\hat{x}_{i} (k)} \\  \, \, \,\, \, \, \,\,\, \, \, \, \, {=A(\bar{x}_{i} (k)+K_{i} (k)(\beta _{i} (k)y_{i} (k)+(1-\beta _{i} (k))C_{i} m_{i} (k)} \\ \, \, \,\, \, \, \,\,\, \, \quad {-C_{i} \bar{x}_{i} (k))\, +\gamma _{i} \sum _{j\in N_{i} }\sigma _{i,j} (k)\beta _{j} (k)(\tilde{x}_{j} (k)-\tilde{x}_{i}  (k))),} \end{array} 
	\end{equation} 
	Substituting \eqref{ZEqnNum466700} into \eqref{ZEqnNum429555} and using \eqref{ZEqnNum250987}, the state estimation error dynamics becomes 
	\begin{equation} \label{ZEqnNum162461} 
		\begin{array}{l} {\bar{\eta }_{i} (k+1)=AM_{i} (k)\bar{\eta }_{i} (k)+A\gamma _{i} \sum _{j\in N_{i} }a_{ij} (k)(\tilde{\eta }_{j} (k)-\tilde{\eta }_{i} (k) )} \\ {\, \, \, \, \, \, \, \, \, \, \, \, \, \, \, \, \, \, \, \, \, \, \, \,\, \, \, \, \, \,\, \,\, \, \, \, \, \, \, \, \,\, \, -AK_{i} (k)(1-\beta _{i} (k))C_{i} \varepsilon _{i} (k),} \end{array} 
	\end{equation} 
	where $a_{ij} (k)=\sigma _{i,j} (k)\beta _{j} (k)$ and $M_{i} (k)=I-K_{i} (k)C_{i} $. 
	
	\noindent Using \eqref{ZEqnNum162461} and notations defined in \eqref{ZEqnNum500695}, the global form of error dynamics becomes
	\begin{equation} \label{ZEqnNum454905} 
		\begin{array}{l} {\bar{\eta }(k+1)=(I_{N} \otimes A)M(k)\bar{\eta }(k)-(\Upsilon \otimes A){\rm {\mathcal L}}(k)\tilde{\eta }(k)} \\ \, \, \, \, \, \, \, \, \, \, \, \, \, \, \, \, \, \, \, \, \, \, \, \,\, \, \, \, \, \,\, \,\, \,\, \, \, \,\,\, \, \,{-(\bar{\beta }\otimes A)(I_{nN} -M(k))E(k)).} \end{array} 
	\end{equation} 
	
	Note that Assumption 4 implies that the total number of the compromised sensors is less than half of the total number of sensors in the network. That is, if $q$ neighbors of an intact sensor node are attacked and collude to send the same value to mislead it, there still exists $q+1$ intact neighbors that communicate values different from the compromised ones. Moreover, since at least half of the intact sensor's neighbors are intact, it can update its beliefs to discard the compromised neighbor's state estimates. Furthermore, since the time-varying graph ${\rm {\mathcal G}}(k)$ resulting from isolating the compromised sensors, based on Assumptions 3 and 4, the entire network is still collectively observable. Using the trust and confidence of neighboring sensors, the incoming information from the compromised communication channels is discarded. 
	
	Now taking norm of equation \eqref{ZEqnNum454905} from both sides and then using the triangular inequality, one has
	\begin{equation} \label{ZEqnNum800097} 
		\begin{array}{l} {\left\| \bar{\eta }(k+1)\right\| \le \left\| (I_{N} \otimes A)M(k)\bar{\eta }(k)\right\| +\left\| (\Upsilon \otimes A){\rm {\mathcal L}}(k)\tilde{\eta }(k)\right\| } \\\, \, \, \, \, \, \, \, \, \, \, \, \, \, \, \, \, \, \, \, \, \, \, \,\, \, \, \, \, \,\, \,\, \, \, \, \, \, \, \, \,\, \, {+\left\| (\bar{\beta }\otimes A)(I_{nN} -M(k))E(k)\right\| .} \end{array} 
	\end{equation} 
	Using \eqref{ZEqnNum466700}, \eqref{ZEqnNum800097} can be rewritten as
	\begin{equation} \label{ZEqnNum810116} 
		\begin{array}{l} {\left\| \bar{\eta }(k+1)\right\| \le A_{o} \left\| \bar{\eta }(k)\right\| +\sigma _{\max } ({\rm {\mathcal L}}(k))\left\| (\Upsilon \otimes A)\tilde{\eta }(k)\right\| } \\ \, \, \, \, \, \, \, \, \, \, \, \, \, \, \, \, \, \, \, \, \, \, \, \,\, \, \, \, \, \,\, \,\, \, \, \, \, \, \, \, \,\, \, {+\left\| (\bar{\beta }\otimes A)-(\bar{\beta }\otimes I_{n} )(I_{N} \otimes A)M(k))E(k)\right\| .} \end{array} 
	\end{equation} 
	{After some manipulations, equation \eqref{ZEqnNum810116} becomes} 
	\begin{equation} \label{ZEqnNum297239} 
		\begin{array}{l} {\left\| \bar{\eta }(k+1)\right\| \le A_{o} \left\| \bar{\eta }(k)\right\| +\sigma _{\max } (A)\sigma _{\max } ({\rm {\mathcal L}}(k))\Upsilon _{m} \left\| \tilde{\eta }(k)\right\| } \\ {\, \, \, \, \, \, \, \, \, \, \, \, \, \, \, \, \, \, \, \, \, \, \, \,\, \, \, \, \, \,\, \,\, \, \, \, \, \, \, \, \,\, \, +(\sigma _{\max } (A)+\sigma _{\max } (A_{o} ))\left\| \bar{\beta }\right\| \sqrt{N} \tau ,} \end{array} 
	\end{equation} 
	with $\Upsilon _{m}$ defined in \eqref{ZEqnNum500695}. Then, using \eqref{ZEqnNum528573}, one can write  \eqref{ZEqnNum297239} as 
	\begin{equation} \label{ZEqnNum560131} 
		\begin{array}{l} {\left\| \bar{\eta }(k+1)\right\| \le A_{o} \left\| \bar{\eta }(k)\right\|  +(\sigma _{\max } (A)+\sigma _{\max } (A_{o} ))\left\| \bar{\beta }\right\| \sqrt{N} \tau} \\ {\, \, \, \, \, \, \, \, \, \, \, \, \, \, \, \, \, \, \, \, \, \, \, \,\, \, \, \, \, \,\, \,\, \, \, \, \, \, \, \, \,\, \,+\sigma _{\max } (A)\sigma _{\max } ({\rm {\mathcal L}}(k))\Upsilon _{m} \sqrt{N(} \alpha /\left\| C_{i} \right\| +{\rm {\mathcal B}})},  \end{array} 
	\end{equation} 
	After solving \eqref{ZEqnNum560131}, one has
	\vspace{-0.2cm}
	\begin{equation} \label{ZEqnNum925065} 
		\left\| \bar{\eta }(k)\right\| \le (A_{o} )^{k} \left\| \bar{\eta }(0)\right\| +\sum _{m=0}^{k-1}(A_{o} )^{k-m-1}  B_{o} , 
	\end{equation} 
	where $A_{0}$ and $B_{0}$ are given in \eqref{ZEqnNum594295}. This completes the proof of part 1.  Based on Assumption 3, the distributed sensor network is always collectively observable. Thus, based on result provided in \cite{d6}, one can conclude that $A_{0}$ in \eqref{ZEqnNum925065} is always Schur and then the upper bound on state estimation error becomes \eqref{64)}. This completes the proof.
\end{proof}
\vspace{-0.15cm}

\noindent { \textbf{Remark 7.} 
	To recap, Theorems 1-3 aim to provide us theoretical analysis to show the vulnerability of event-triggered DKF mechanism to deception attack consist of reply attack and false data injection attack. Moreover, Theorems 4-5 aim to build a mechanism to detect these types of attacks on event-triggered DKF and to mitigate the effects of them. To this aim, the results of Theorems 4 and 5 are essential to develop an entropy estimation-based attack detection approach for the event-triggered DKF, and Theorem 6 and corresponding Algorithm 3 complete the machinery required for mitigation scheme by estimating the actual state based on the attack detection approach presented in Algorithms 1 and 2. }

\begin{algorithm}[!ht]
	\caption{Secure Distributed Estimation Mechanism (SDEM).}
		\begin{enumerate}
			\item [1:] Start with initial innovation sequences and design parameters $\Upsilon _{1} $ and $\Lambda _{1}$.
			\item [2:] \textbf{procedure $\forall i=1,\ldots ,N$ }
			\item [3:] Use samples of innovation sequences $\{ r_{i}^{a} (l),\ldots, $ \qquad $r_{i}^{a} (l-1+w)\}$ and $\{ r_{i} (l),\ldots ,r_{i} (l-1+w)\} $ based on \eqref{ZEqnNum255093} and \eqref{ZEqnNum276515}, $\forall l\in k.$
			\item [4:] Estimate the $\hat{D}_{KL} (P_{r_{i}^{a} } ||P_{r_{i} } )$ using \eqref{ZEqnNum750552}.
			\item [5:] Based on \eqref{ZEqnNum125869}-\eqref{ZEqnNum359584}, compute confidence $\beta _{i} (k)$ as
			\begin{equation}\label{Alg1}
				\beta _{i} (k)=\Upsilon _{1} \sum _{l=0}^{k-1}\frac{(\kappa _{1} )^{k-l+1} }{\Upsilon _{1} +\hat{D}_{KL} (P_{r_{i}^{a} } ||P_{r_{i} } )}.
			\end{equation}  
			\item [6:] For each sensor node $j\in N_{i} $, use samples of innovation sequences $\{ \zeta _{i,j}^{a} (l),\ldots ,\zeta _{i,j}^{a} (l-1+w)\}$ and $\{ r_{i} (l),\ldots ,r_{i} (l-1+w)\}$ based on \eqref{ZEqnNum178443} and \eqref{ZEqnNum276515}, $\forall l\in k.$ 
			\item [7:] Estimate the $\hat{D}_{KL} (P_{\zeta _{i,j}^{a} } ||P_{r_{i} } )$ using \eqref{ZEqnNum691139}.
			\item [8:] Using \eqref{ZEqnNum846884}-\eqref{ZEqnNum805360}, compute trust $\sigma _{i,j}(k)$ as
			\begin{equation}\label{Alg2}
				\sigma _{i,j} (k)=\Lambda _{1} \sum _{l=0}^{k-1}\frac{(\kappa _{2} )^{k-l+1} }{\Lambda _{1} +\hat{D}_{KL} (P_{\zeta _{i,j}^{a} } ||P_{r_{i} } )}  \theta _{i,j} (l).
			\end{equation}
			\item [9:] Using the sensor measurement $y_{i} (k)$ with the confidence $\beta _{i} (k)$ {in \eqref{Alg1}}, the trust on neighbor's $\sigma _{i,j} (k)$ {in \eqref{Alg2}} and neighbor's state estimates $\tilde{x}_{j} (k),\, \, \forall j\in N_{i} $, update the resilient state estimator in \eqref{ZEqnNum565391}.
			\item [10:] \textbf{end procedure}
		\end{enumerate}
\end{algorithm}

\vspace{-0.2cm}

\section{ Simulation Results}
In this section, we discuss simulation results to demonstrate the efficacy of presented attack detection and mitigation mechanism.

{Consider the following simple longitudinal-direction cruise dynamics of an autonomous underwater vehicle (AUV)
	\begin{align}
		{ \left[ {\begin{array}{*{20}{c}}
					{x(k + 1)}\\
					{v(k + 1)}
			\end{array}} \right] = \left[ {\begin{array}{*{20}{c}}
					0&1\\
					0&{ - b/m}
			\end{array}} \right]\left[ {\begin{array}{*{20}{c}}
					{x(k)}\\
					{v(k)}
			\end{array}} \right] + \left[ {\begin{array}{*{20}{c}}
					0\\
					{u(k)/m}
			\end{array}} \right] + w(k)}
		\label{72mj}
	\end{align}
	where $x$ is the longitudinal position, $v$ is the velocity in $X$-direction,  $m=1000~Kg$ is mass, $b=50~{{N{\mathop{\rm Sec}\nolimits} } \mathord{\left/
			{\vphantom {{N{\mathop{\rm Sec}\nolimits} } m}} \right.
			\kern-\nulldelimiterspace} m}$ is an coefficient corresponding to friction and hydrodynamic drag, $w$ is  disturbance force (with Gaussian distribution) generated by underwater and tidal currents, and $u(k)= 1050v(k)$ is the force applied by engine.} 

{Now, consider a scenario in which a sensor network installed undersea with the communication graph topology shown in Fig \ref{fig2} to estimate the longitudinal position and velocity of an AUV cruising undersea.}   

\noindent
{The closed-loop dynamical system \eqref{72mj} can be seen as an autonomous exogenous system (exosystem \cite{MJTACpaper}) as
	\begin{equation}
		Z(k + 1) = \underbrace {\left[ {\begin{array}{*{20}{c}}
					0&1\\
					0&1
			\end{array}} \right]}_AZ(k) + w(k)
		\label{72)} 
	\end{equation} 
	where $Z(k) = {\left[ {\begin{array}{*{20}{c}}
				{x(k)}&{v(k)}
		\end{array}} \right]^T}$. Now, let the observation matrix $C_{i} $ in \eqref{ZEqnNum687942}, noise covariances, and initial state, respectively, be chosen as
	\begin{equation} 
		C_{i} =[0\, \, 3;0\, \, 2],\, \, \, \, \, Q=I_{2} ,\, \, \, \, \, R_{i} =\, I_{2} ,\, \, \, \, \, Z_{0} =(0,0)
		\label{721)} 
	\end{equation} 
	As one can see, the pairs $(A,C_{i} )$ are not observable which indicates that each one of these sensors cannot estimate the longitudinal position and velocity of an AUV individually. Note that, however, \eqref{72)} and \eqref{ZEqnNum687942} with $C_{i}$ given in \eqref{721)} are collectively observable.}
\begin{figure}[!ht]
	\begin{center}
		\includegraphics[width=2.3in,height=1.8in]{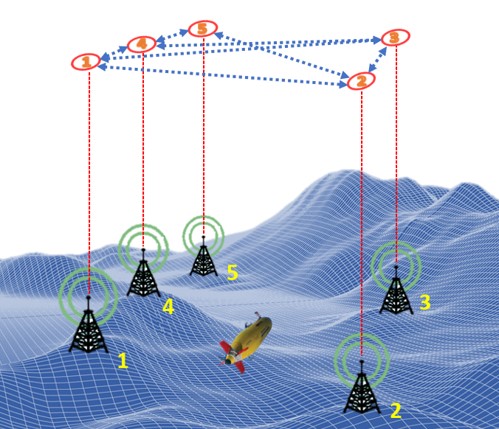}
		\caption{{The scenario of distributed estimating of the position and velocity of an AUV.}}\label{fig2}
		\captionsetup{justification=centering}
	\end{center}
\end{figure}

{For intact sensor network, based on the dynamics \eqref{72)} with covariances given in \eqref{721)}, as depicted in Fig.~\ref{fig:MJ1}, the state estimation errors converge to zero (in the mean square sense) for each sensor node and as the result the state estimations of sensors converge to the true states. Moreover, the event generation based on the event-triggering mechanism in (\ref{eq3x}) with the triggering threshold $\alpha =1.35$ is shown in Fig.~\ref{fig:MJ2}.}

\begin{figure}[!t] 
	\begin{minipage}{0.24\textwidth}
		\includegraphics[width=1.1\linewidth, height=29mm]{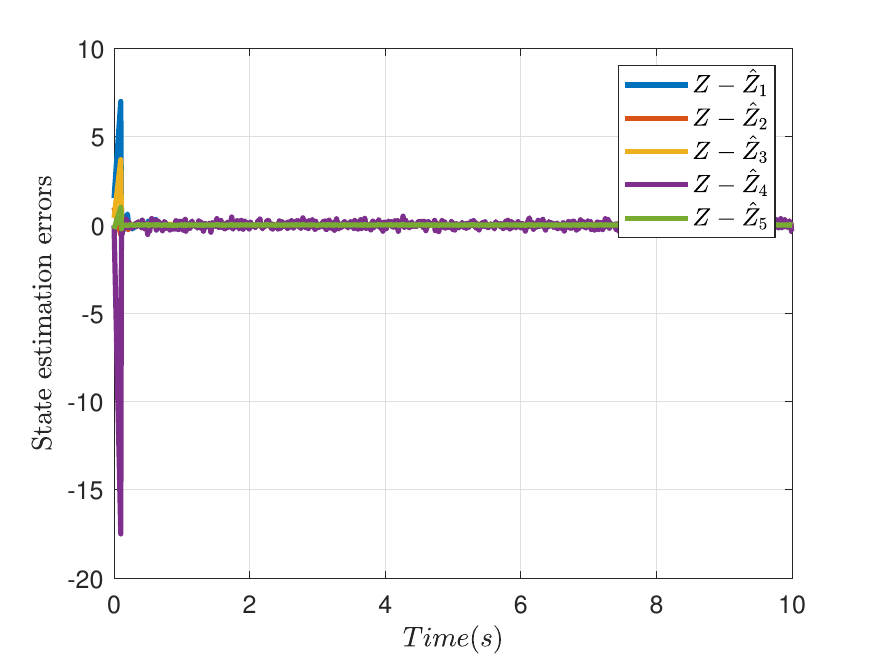} 
		\caption{}
		\label{fig:MJ1}
	\end{minipage}
	\begin{minipage}{0.24\textwidth}
		\includegraphics[width=1.1\linewidth, height=28mm]{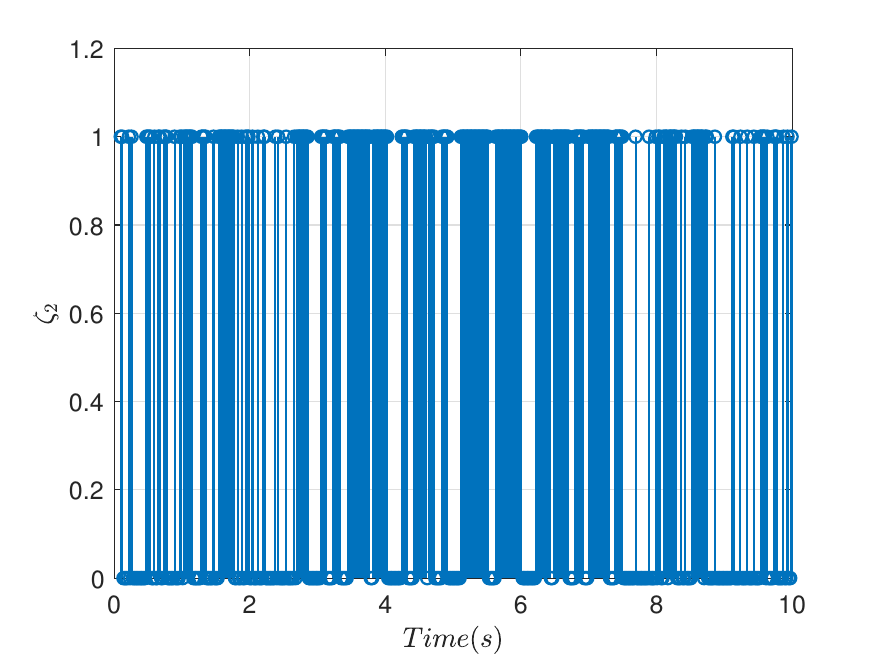}
		\caption{}
		\label{fig:MJ2}
	\end{minipage} 
	\caption{Sensor network without any attack. (a) State estimation errors   (b) Transmit function for sensor node $2$.}
	\vspace{-0.35cm}
\end{figure}
{Now, assume that sensor node $2$ in the network is compromised with the adversarial input $\delta _{2} (k)= 9\sin (100k)$ after the time instance $t = 10~Sec$. 
	Fig.~\ref{fig:MJ3} shows the attacker's effect on sensor node $2$ and one can see that the compromised sensors and other sensors in the network deviates from desired target state and results in non-zero estimation error based on attacker's input.}

{Fig.~\ref{fig:MJ4} illustrates the event generation based on the event-triggering mechanism in (\ref{eq3x}) in the presence of attack. Fig.~\ref{fig:MJ4} shows that after injection of the attack on sensor node $2$, the event-triggered system becomes time-triggered and demonstrates continuous-triggering misbehavior. This result follows the analysis presented for the continuous-triggering misbehavior. The results for non-triggering misbehavior for sensor node $2$ is depicted in Figs.~\ref{fig:MJ5}-\ref{fig:MJ6} which follow the presented analysis.}
\begin{figure}[H] 
	\begin{minipage}{0.24\textwidth}
		\includegraphics[width=1.1\linewidth, height=29mm]{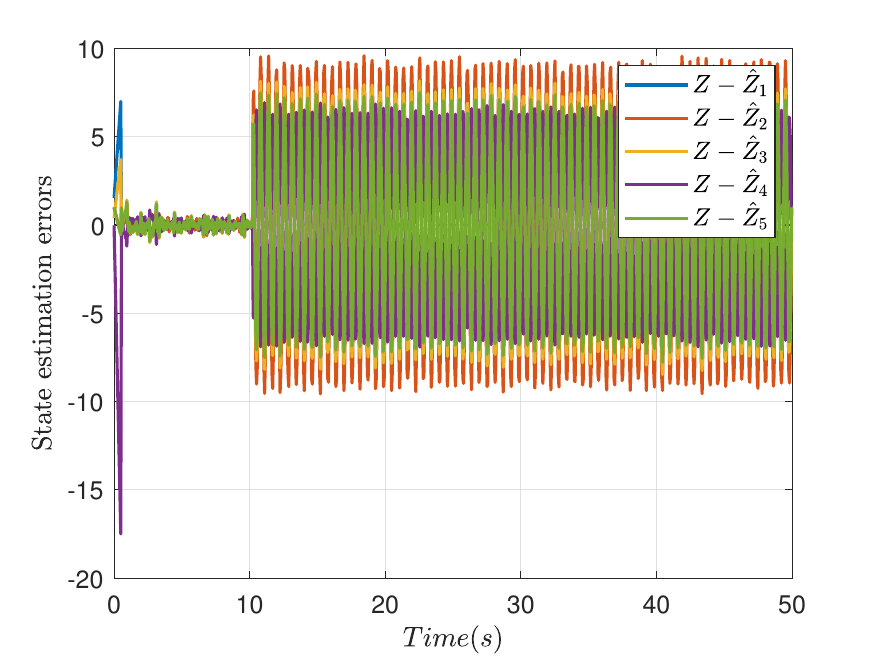} 
		\caption{}
		\label{fig:MJ3}
	\end{minipage}
	\begin{minipage}{0.24\textwidth}
		\includegraphics[width=1.1\linewidth, height=28mm]{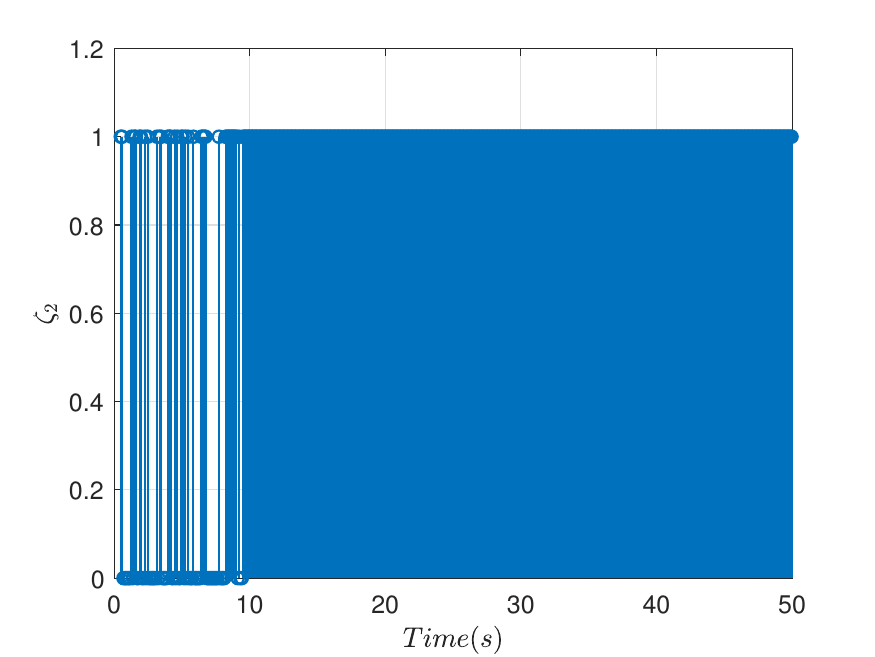}
		\caption{}
		\label{fig:MJ4}
	\end{minipage} 
	\caption{Sensor node $2$ under continuous-triggering misbehavior. (a) State estimation errors  (b) Transmit function for sensor node $2$.}
	\vspace{-0.35cm}
\end{figure}

\noindent
{Using the presented attack detection mechanism, one can detect the effect of the attack on the sensor nodes. Fig.~\ref{fig:MJ7}  illustrates the result for estimated KL divergence-based attack detection mechanism and it shows that after the injection of attack signal into sensor node 2 at $t=10~Sec$ the estimated KL divergence starts increasing for compromised sensor node 2. The estimated divergence for the compromised sensor, i.e., sensor node 2 grows after attack injection at $t=10~Sec$ which follows the result presented in Theorem 4.}
\begin{figure}[H] 
	\begin{minipage}{0.24\textwidth}
		\includegraphics[width=1.1\linewidth, height=29mm]{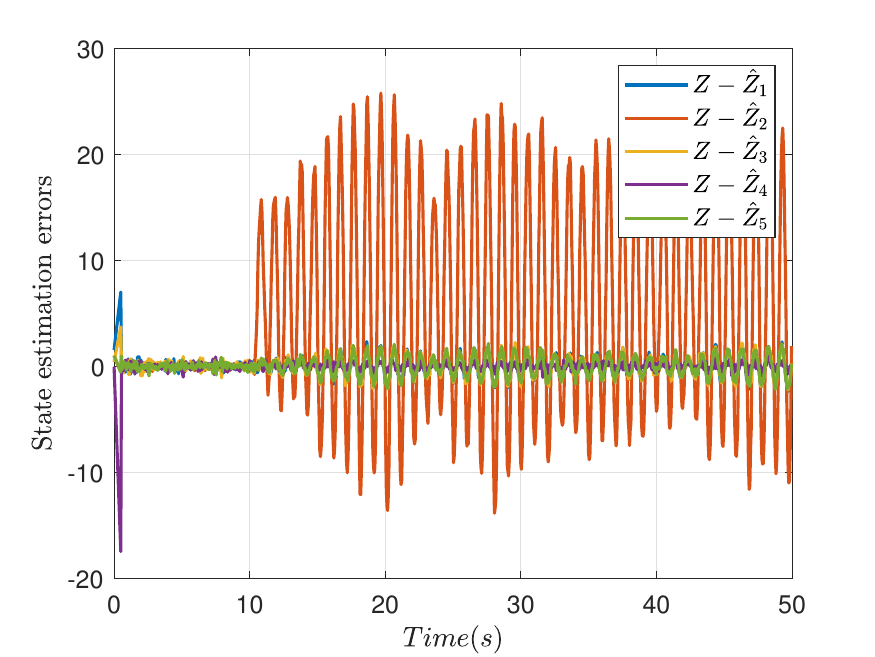} 
		\caption{}
		\label{fig:MJ5}
	\end{minipage}
	\begin{minipage}{0.24\textwidth}
		\includegraphics[width=1.1\linewidth, height=28mm]{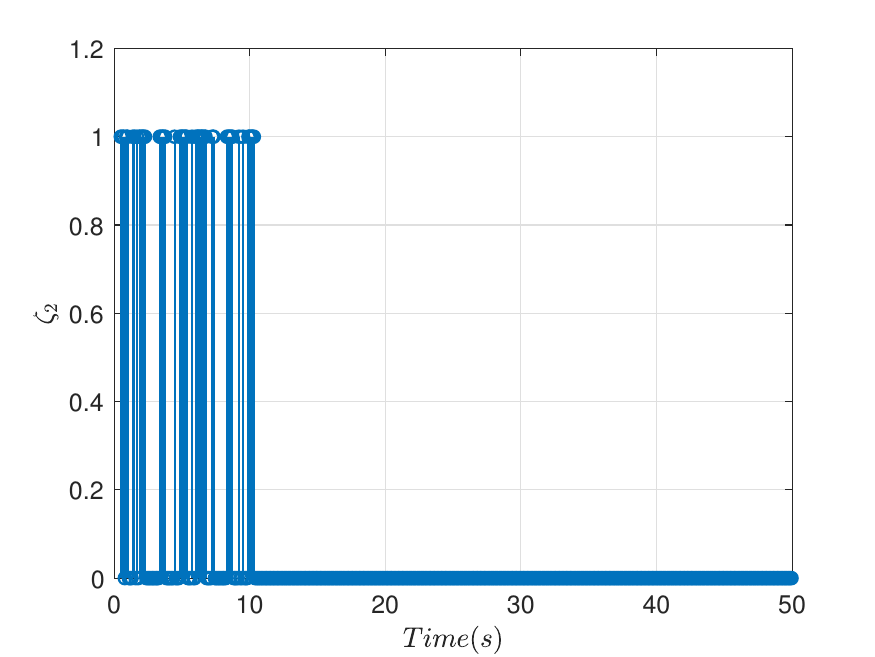}
		\caption{}
		\label{fig:MJ6}
	\end{minipage} 
	\caption{Sensor node $2$ under non-triggering misbehavior. (a) State estimation errors  (b) Transmit function for sensor node $2$.}
	\vspace{-0.35cm}
\end{figure}

\begin{figure}[!t]
	\centering{\includegraphics [width=2.3in] {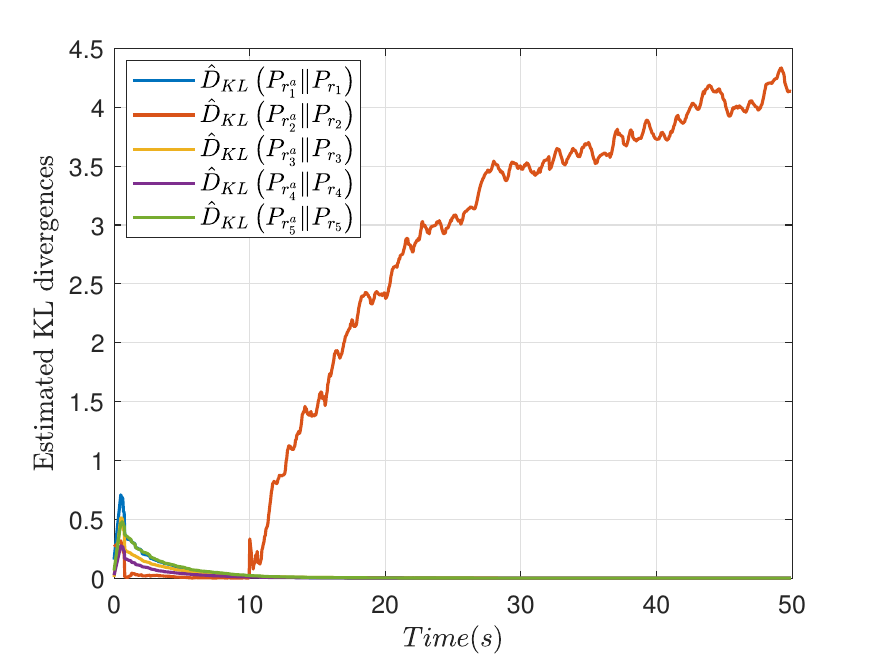}}
	\caption{{Estimated  KL divergences in the case that sensor node $2$ is under non-triggering attack.}}
	\label{fig:MJ7}
\end{figure}

{The confidence of the sensor is evaluated based on the Lemma 1 with the discount factor $\kappa _{1} =0.5$ and the uncertainty threshold as $\Upsilon _{1} =0.8$.  Fig.~\ref{fig:MJ8} shows the confidence of sensors in the presence of the considered attack which is close to one for healthy sensors and tends to zero for the compromised one.   Then, the belief based proposed resilient estimator is implemented and Fig.~\ref{fig:MJ9} shows the result for the state estimation using the resilient estimator \eqref{ZEqnNum565391}. After the injection of attack, within a few seconds, the sensors reach consensus on the state estimates, i.e., the state estimates of sensors converge to the actual position of the target. The result in Fig.~\ref{fig:MJ9} follows Theorem 6. 
}

\begin{figure}[!t]
	\centering{\includegraphics [width=2.3in] {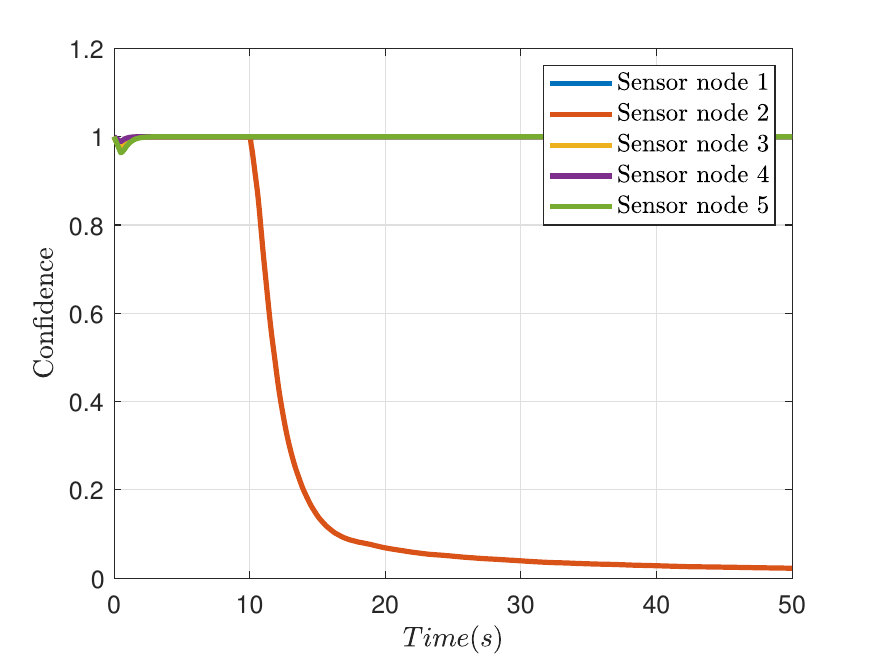}}
	\caption{{Confidence of sensors in the case that sensor node $2$ is under non-triggering attack.}}
	\label{fig:MJ8}
\end{figure}

\begin{figure}[!t]
	\centering{\includegraphics [width=2.3in] {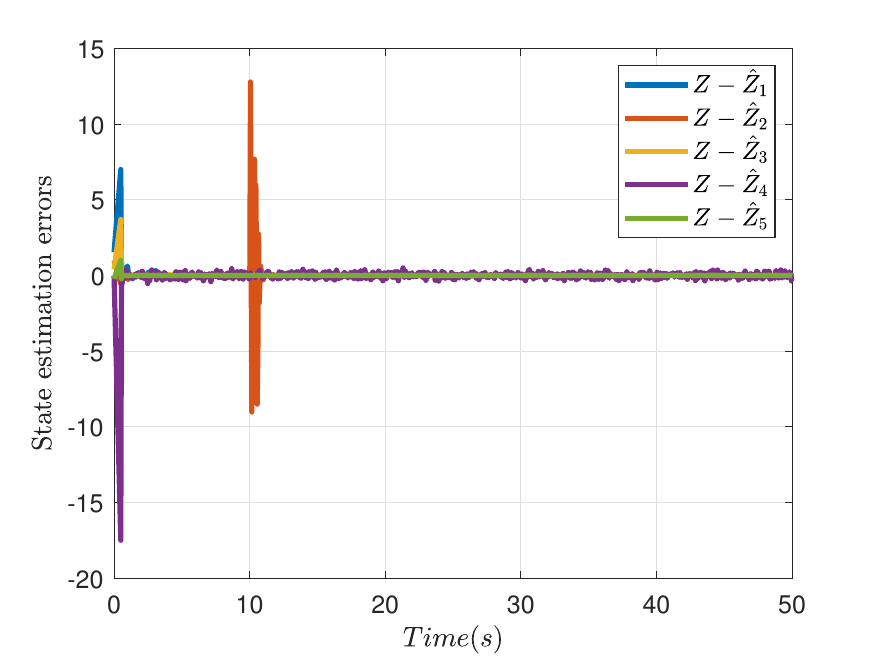}}
	\caption{{State estimation errors under attack on sensor node $2$ using proposed resilient state estimator.}}
	\label{fig:MJ9}
\end{figure}

\vspace{-0.25cm}

\section{ Conclusion}

In this paper, first, we analyze the adverse effects of cyber-physical attacks on the event-triggered distributed Kalman filter (DKF). We show that attacker can adversely affect the performance of the DKF. We also show that the event-triggered mechanism in the DKF can be leveraged by the attacker to result in a non-triggering misbehavior that significantly harms the network connectivity and its collective observability. Then, {to detect adversarial intrusions in the DKF, we relax restrictive Gaussian assumption on probability density functions of attack signals and estimate the Kullback-Leibler (KL) divergence via  $k$-nearest neighbors approach. }Finally, to mitigate attacks, a meta-Bayesian approach is presented that incorporates the outcome of the attack detection mechanism to perform second-order inference and consequently form beliefs over beliefs, i.e., confidence and trust of a sensor. Each sensor communicates its confidence to its neighbors. Sensors then incorporate the confidence of their neighbors and their own trust about their neighbors into their posterior update laws to successfully discard corrupted sensor information. Then, the simulation result illustrates the performance of the presented resilient event-triggered DKF. {Future research will focus on addressing the effect of accuracy of the proposed attack detection mechanism on the proposed mitigation mechanism.}

\vspace{-0.3cm}
\appendices
\section{Proof of Theorem 1}
Note that for the notional simplicity, in the following proof, we keep the sensor index $i$ but ignore the time-indexing $k$. Without the time index, we represent the prior at time $k+1$ as $\bar{x}_{i}^{a} (k+1)\buildrel\Delta\over= (\bar{x}_{i}^{a} )^{+} $ and follow the same for other variables.

Using the process dynamics in \eqref{ZEqnNum820040} and the corrupted prior state estimate in \eqref{ZEqnNum120276}, one has
\vspace{-0.12cm}
\begin{equation} \label{ZEqnNum533668} 
	(\bar{\eta }_{i}^{a} )^{+} =x^{+} -(\bar{x}_{i}^{a} )^{+} =A(x\, -\hat{x}_{i}^{a} )\, +\, w, 
\end{equation} 
where the compromised posterior state estimate $\hat{x}_{i}^{a} (k)$ follows the dynamics \eqref{ZEqnNum120276}. Similarly, using \eqref{ZEqnNum120276}, the corrupted posterior state estimation error becomes
\vspace{-0.12cm}
\begin{equation} \label{ZEqnNum240480} 
	\eta _{i}^{a} =x-\hat{x}_{i}^{a} =x-\bar{x}_{i}^{a} -K_{i}^{a} (y_{i} -C\bar{x}_{i}^{a} )-\gamma _{i} \sum _{j\in N_{i} }(\tilde{x}_{j}^{a} -\tilde{x}_{i}^{a}  )-K_{i}^{a} f_{i} . 
\end{equation} 
\vspace{-0.12cm}
Then, one can write \eqref{ZEqnNum533668}-\eqref{ZEqnNum240480} as
\begin{equation} \label{ZEqnNum404232} 
	\left\{\begin{array}{l} {(\bar{\eta }_{i}^{a} )^{+} =A\eta _{i}^{a} \, +\, w,} \\ {\eta _{i}^{a} =(I_{n} -K_{i}^{a} C_{i} )\bar{\eta }_{i}^{a} -K_{i}^{a} v_{i} +u_{i}^{a} ,} \end{array}\right.  
\end{equation} 
where
\begin{equation} \label{ZEqnNum571757} 
	\vspace{-0.12cm}
	u_{i}^{a} =\gamma _{i} \sum _{j\in N_{i} }(\tilde{\eta }_{j}^{a}  -\tilde{\eta }_{i}^{a} )-K_{i}^{a} f_{i} . 
\end{equation} 

Based on \eqref{ZEqnNum926700}, we define the predictive state estimation error, respectively, under attack as
\vspace{-0.12cm}
\begin{equation} \label{ZEqnNum896894} 
	\begin{array}{l}(\tilde{\eta }_{i}^{a} )^{+} =x^{+} -(\tilde{x}_{i}^{a} )^{+}\\ \, \, \, \, \,\, \,  \, \, \, \,\, \, \, \, \, \,\,\, \, \,  \,\,\, =\zeta _{i}^{+} (\bar{\eta }_{i}^{a} )^{+} +(1-\zeta _{i}^{+} )A\tilde{\eta }_{i}^{a} . \end{array} 
\end{equation} 

Using \eqref{ZEqnNum404232}, the corrupted covariance of the prior state estimation error becomes
\begin{equation} \label{ZEqnNum928831} 
	\begin{array}{l} {(\bar{P}_{i}^{a} )^{+} ={\bf {\rm E}}\left[(\bar{\eta }_{i}^{a} )^{+} ((\bar{\eta }_{i}^{a} )^{+} )^{T} \right],} \\ {\, \, \, \, \,\, \,  \, \, \, \, \,\,\, \, \, \, \, \, \, \, \, \, \, ={\bf {\rm E}}\left[(A\eta _{i}^{a} \, +\, w)(A\eta _{i}^{a} \, +\, w\, )^{T} \right]=A\hat{P}_{i}^{a} A^{T} +Q.} \end{array} 
\end{equation} 

Using the corrupted predictive state estimate error $\, (\tilde{\eta }_{i}^{a} )^{+} $ in \eqref{ZEqnNum896894} with $(\bar{P}_{i,j}^{a} )^{+} =A\hat{P}_{i,j}^{a} A^{T} +Q$, one can write the cross-correlated predictive state estimation error covariance $(\tilde{P}_{i,j}^{a} )^{+} $ as
\vspace{-0.12cm}
\begin{equation} \label{ZEqnNum354214} 
	\begin{array}{l} {(\tilde{P}_{i,j}^{a} )^{+} ={\bf {\rm E}}\left[(\tilde{\eta }_{i}^{a} )^{+} ((\tilde{\eta }_{j}^{a} )^{+} )^{T} \right]} \\ \,\,\,\,\,\,\,\,\,\,\,\,\,\,\,\,\,\,\,\,\,\,\,\,\,\,{=\zeta _{i}^{+} (1-\zeta _{j}^{+} )A(\breve{P}_{i,j}^{a} )^{+} +(1-\zeta _{i}^{+} )\zeta _{j}^{+} (\stackrel{\frown}{P}_{i,j}^{a} )^{+} A^{T} } \\ \,\,\,\,\,\,\,\,\,\,\,\,\,\,\,\,\,\,\,\,\,\,\,\,\,\,\,\,\,\,\,\,\,\,{+\zeta _{i}^{+} \zeta _{j}^{+} (\bar{P}_{i,j}^{a} )^{+} +(1-\zeta _{i}^{+} )(1-\zeta _{j}^{+} )(A\tilde{P}_{i,j}^{a} A^{T} +Q),} \end{array} 
\end{equation} 
where $\stackrel{\frown}{P}_{i,j}^{a} $ and $\breve{P}_{i,j}^{a}$ be the cross-correlated estimation error covariances and their updates are given in \eqref{ZEqnNum358063}-\eqref{ZEqnNum655968}. 

The cross-correlated estimation error covariance $(\stackrel{\frown}{{P}}_{i,j}^{a} )^{+}$ in \eqref{ZEqnNum354214} is given by
\vspace{-0.12cm}
\begin{equation} \label{ZEqnNum358063} 
	\begin{array}{l} {(\stackrel{\frown}{P}_{i,j}^{a} )^{+} ={\bf {\rm E}}\left[(\tilde{\eta }_{i}^{a} )^{+} ((\bar{\eta }_{j}^{a} )^{+} )^{T} \right]} \\ \,\,\,\,\,\,\,\,\,\,\,\,\,\,\,\,\,\,\,\,\,\,\,\,\,\,\,\,\,\,{=\zeta _{i}^{+} (\bar{P}_{i,j}^{a} )^{+} +(1-\zeta _{i}^{+} )A\sum _{r\in N_{i} }(\tilde{P}_{i,r}^{a} -\tilde{P}_{i,j}^{a} )(\gamma _{i}  A)^{T} +} \\ {\,\,\,\,\,\,\,\,\,\,\,\,\,\,\,\,\,\,\,\,\,\,\,\,\,\,\,\,\,\,\,\,\,\,\,\,\,\,\,  (1-\zeta _{i}^{+} )[A\stackrel{\frown}{P}_{i,j}^{a} M_{i}^{a} A^{T} +Q],} \end{array} 
\end{equation} 
where $\tilde{P}_{i,j}^{a}$ and $\breve{P}_{i,j}^{a}$ denote the cross-correlated estimation error covariances evolve according to \eqref{ZEqnNum354214} and \eqref{ZEqnNum655968}. Similarly, $(\breve{P}_{i,j}^{a} )^{+} $ is updated based on the expression given by
\begin{equation} \label{ZEqnNum655968} 
	\begin{array}{l} {(\breve{P}_{i,j}^{a} )^{+} ={\bf {\rm E}}\left[(\bar{\eta }_{i}^{a} )^{+} ((\tilde{\eta }_{j}^{a} )^{+} )^{T} \right]} \\ \,\,\,\,\,\,\,\,\,\,\,\,\,\,\,\,\,\,\,\,\,\,\,\,\,\,\,{={\bf {\rm E}}\left[(\bar{\eta }_{i}^{a} )^{+} (\zeta _{j}^{+} (\bar{\eta }_{j}^{a} )^{+} +(1-\zeta _{j}^{+} )(A\tilde{\eta }_{j}^{a} +w)\, )^{T} \right]} \\ \,\,\,\,\,\,\,\,\,\,\,\,\,\,\,\,\,\,\,\,\,\,\,\,\,\,\,{=\zeta _{j}^{+} (\bar{P}_{i,j}^{a} )^{+} +(1-\zeta _{j}^{+} )[A(M_{i}^{a} )^{T} \stackrel{\frown}{P}_{i,j}^{a} A^{T} +Q]} \\ \,\,\,\,\,\,\,\,\,\,\,\,\,\,\,\,\,\,\,\,\,\,\,\,\,\,\,\,\,\,\,\,\,\,\,{+(1-\zeta _{j}^{+} )A\gamma _{i} \sum _{s\in N_{i} }(\tilde{P}_{s,j}^{a} -\tilde{P}_{i,j}^{a} ) A^{T} .} \end{array} 
\end{equation} 

Now using \eqref{ZEqnNum240480}-\eqref{ZEqnNum896894}, one can write the covariance of posterior estimation error  $\hat{P}_{i}^{a} $ as
\begin{equation} \label{ZEqnNum893608} 
	\begin{array}{l} {\hat{P}_{i}^{a} ={\bf {\rm E}}[M_{i} \bar{\eta }_{i}^{a} (M_{i} \bar{\eta }_{i}^{a} )^{T} ]+{\bf {\rm E}}[K_{i}^{a} v_{i} (K_{i}^{a} v_{i} )^{T} ]}  -2{\bf {\rm E}}[(M_{i} \bar{\eta }_{i}^{a} )(K_{i}^{a} v_{i} )^{T} ]\\ 
			\,\,\,\,\,\,\,\,\,\,\,\,\,-2{\bf {\rm E}}[K_{i}^{a} v_{i} (\gamma _{i} u_{i}^{a} )^{T} ]  {+{\bf {\rm E}}[(\gamma _{i} u_{i}^{a} )(\gamma _{i} u_{i}^{a} )^{T} ]+2{\bf {\rm E}}[(M_{i} \bar{\eta }_{i}^{a} )(\gamma _{i} u_{i}^{a} )^{T} ],} \end{array} 
\end{equation} 
Using \eqref{ZEqnNum928831} and measurement noise covariance, the first two terms of \eqref{ZEqnNum893608} become
\begin{equation} \label{87)} 
	\begin{array}{l} {{\bf {\rm E}}[M_{i} \bar{\eta }_{i}^{a} (M_{i} \bar{\eta }_{i}^{a} )^{T} ]=M_{i} \bar{P}_{i}^{a} M_{i}^{T},}\,\,\,\, \,\, {{\bf {\rm E}}[K_{i}^{a} v_{i} (K_{i}^{a} v_{i} )^{T} ]=K_{i}^{a} R_{i} (K_{i}^{a} )^{T}_.} \end{array} 
\end{equation} 
According to Assumption 1, the measurement noise $v_{i} $ is i.i.d. and uncorrelated with state estimation errors, therefore, the third and fourth terms in \eqref{ZEqnNum893608} become zero.  Now $u_{i}^{a} $ in \eqref{ZEqnNum571757} and Assumption 1, the last two terms in \eqref{ZEqnNum893608} can be simplified as 
\vspace{-0.2cm}
\begin{equation} \label{88)} 
	\begin{array}{l} {{\bf {\rm E}}[(u_{i}^{a} )(u_{i}^{a} )^{T} ]} =\gamma _{i} {}^{2} ({\bf {\rm E}}\left[[\sum _{j\in N_{i} }(\tilde{\eta }_{j}^{a}  -\tilde{\eta }_{i}^{a} )][\sum _{j\in N_{i} }(\tilde{\eta }_{j}^{a}  -\tilde{\eta }_{i}^{a} ))]^{T} \right]\\ \,\,\,\,\,\,\,\,\,\,\,\,\,\,\,\,\,\,\,\,\,\,\,\,\,\,\,\,\,\,\,\,\,\,\,\,\,\,\,\,\,\,\,\,\,\,\,\,\,+{\bf {\rm E}}[K_{i}^{a} f_{i} (K_{i}^{a} f_{i} )^{T} ] {-2K_{i}^{a} {\bf {\rm E}}[f_{i} \sum _{j\in N_{i} }(\tilde{\eta }_{j}^{a}  -\tilde{\eta }_{i}^{a} )^{T} ]),} \\ \,\,\,\,\,\,\,\,\,\,\,\,\,\,\,\,\,\,\,\,\,\,\,\,\,\,\,\,\,\,\,\,\,\,\,\,\,\,\,\,\,\,\,\,\,\,=\gamma _{i} {}^{2} (\sum _{j\in N_{i} }(\tilde{P}_{j}^{a}  -2\tilde{P}_{i,j}^{a} +\tilde{P}_{i}^{a} )+K_{i}^{a} \Sigma _{i}^{f} (K_{i}^{a} )^{T}\\ \,\,\,\,\,\,\,\,\,\,\,\,\,\,\,\,\,\,\,\,\,\,\,\,\,\,\,\,\,\,\,\,\,\,\,\,\,\,\,\,\,\,\,\,\,\,\,\,\,\,\,\,\, -2K_{i}^{a} {\bf {\rm E}}[f_{i} \sum _{j\in N_{i} }(\tilde{\eta }_{j}^{a}  -\tilde{\eta }_{i}^{a} )^{T} ]), \end{array} 
\end{equation} 
\vspace{-0.1cm}
and 
\begin{equation} \label{ZEqnNum612155} 
	\begin{array}{l} {2{\bf {\rm E}}[(u_{i}^{a} )(M_{i} \bar{\eta }_{i}^{a} )^{T} ]=2{\bf {\rm E}}[(\gamma _{i} \sum _{j\in N_{i} }(\tilde{\eta }_{j}^{a}  -\tilde{\eta }_{i}^{a} )-K_{i}^{a} f_{i} )(M_{i} \bar{\eta }_{i}^{a} )^{T} ],} \\ \,\,\,\,\,\,\,\,\,\,\,\,\,\,\,\,\,\,\,\,\,\,\,\,\,\,\,\,\,\,\,\,\,\,\,\,\,\,\,\,\,\,\,\,\,\,{=2\gamma _{i} \sum _{j\in N_{i} }(\stackrel{\frown}{P}_{i,j}^{a}  -\stackrel{\frown}{P}_{i}^{a} )M_{i}^{T} -2K_{i}^{a} {\bf {\rm E}}[f_{i} (\bar{\eta }_{i}^{a} )^{T} ]M_{i}^{T},} \end{array} 
\end{equation} 
where the cross-correlated term $\stackrel{\frown}{P}_{i,j}^{a} $ is updated according to \eqref{ZEqnNum358063}. Using \eqref{ZEqnNum893608}-\eqref{ZEqnNum612155}, the posterior state estimation error $P_{i}^{a} $ under attacks is given by
\vspace{-0.2cm}
\begin{equation} \label{90)}\nonumber
	\begin{array}{l} {\hat{P}_{i}^{a} =M_{i}^{a} \bar{P}_{i}^{a} (M_{i}^{a} )^{T} +K_{i}^{a} [R_{i} +\Sigma _{i}^{f} ](K_{i}^{a} )^{T} -2K_{i}^{a} \Xi _{f} } \\ \,\,\,\,\,\,\,\,\,\,\,{+2\gamma _{i} \sum _{j\in N_{i} }(\stackrel{\frown}{P}_{i,j}^{a}  -\stackrel{\frown}{P}_{i}^{a} )(M_{i}^{a} )^{T} +\gamma _{i} {}^{2} (\sum _{j\in N_{i} }(\tilde{P}_{j}^{a}  -2\tilde{P}_{i,j}^{a} +\tilde{P}_{i}^{a} ),} \end{array} 
\end{equation} 
with $\Xi _{f} =[{\bf {\rm E}}[f_{i} \sum _{j\in N_{i} }(\tilde{\eta }_{j}^{a} -\tilde{\eta }_{i}^{a} )^{T}  ])+{\bf {\rm E}}[f_{i} (\bar{\eta }_{i}^{a} )^{T} ](M_{i}^{a} )^{T} ].$ This completes the proof.

\end{document}